\documentclass[10pt]{scrartcl}
\pdfoutput=1
\usepackage[total={16.2cm, 24cm}]{geometry}
\usepackage{microtype}
\usepackage{amsmath,amsthm,amssymb}
\usepackage{tikz}
\usetikzlibrary{patterns,decorations.pathreplacing}
\usepackage[pagebackref]{hyperref}
\hypersetup{
	pdfencoding=auto, 
	psdextra,
	colorlinks=true,
	citecolor=green!50!black,
	linkcolor=red!60!black,
	urlcolor=blue!90!black
}

\usepackage{nicefrac}
\usepackage{natbib}
\usepackage{paralist}
\usepackage{xcolor}
\usepackage{colortbl}
\usepackage{subcaption}
\usepackage{booktabs}
\usepackage{eurosym}
\usepackage[linecolor=green!70!black, backgroundcolor=green!10, bordercolor=black, textsize=small]{todonotes}
\definecolor{winered}{rgb}{0.6,0.1,0.1}

\renewcommand*{\le}{\leqslant}
\renewcommand*{\leq}{\leqslant}
\renewcommand*{\ge}{\geqslant}
\renewcommand*{\geq}{\geqslant}
\renewcommand{\epsilon}{\varepsilon}

\newcommand{\myemph}[1]{{\color{winered}\emph{#1}}}
\newcommand{\naturals}{{{\mathbb{N}}}}
\newcommand{\reals}{{{\mathbb{R}}}}
\newcommand{\rationals}{{{\mathbb{Q}}}}
\newcommand{\cost}{{{\mathrm{cost}}}}

\newcommand{\calR}{{{\mathcal{R}}}}

\theoremstyle{definition}
\newtheorem{definition}{Definition}
\newtheorem*{definition*}{Definition}
\newtheorem{example}{Example}
\newtheorem{remark}{Remark}
\theoremstyle{plain}
\newtheorem{theorem}{Theorem}
\newtheorem*{theorem*}{Theorem}

\newtheorem{lemma}{Lemma}
\newtheorem{proposition}{Proposition}

\makeatletter
\newtheorem*{rep@theorem}{\rep@title}
\newcommand{\newreptheorem}[2]{\newenvironment{rep#1}[1]{\def\rep@title{#2 \ref{##1}}\begin{rep@theorem}}{\end{rep@theorem}}}
\makeatother
\newreptheorem{theorem}{Theorem}
\newreptheorem{proposition}{Proposition}
\newreptheorem{lemma}{Lemma}

\renewenvironment{description}
{\list{}{\labelwidth=10pt \leftmargin=15pt
		}}
{\endlist}

\newcommand{\pos}{{{{\mathrm{pos}}}}}

\newcommand{\argmin}{{\mathrm{argmin}}}

\usepackage[ruled, linesnumbered]{algorithm2e} 
\SetAlFnt{\small}
\SetAlCapFnt{\small}
\SetAlCapNameFnt{\small}
\SetAlCapHSkip{0pt}
\IncMargin{-\parindent}

\usepackage{todonotes} 

\usepackage{tcolorbox}
\tcbuselibrary{breakable,theorems,skins}
\newtcolorbox{examplebox}{
	blank,
breakable,
	left=3.5mm, right=7pt,
	top=6pt, bottom=6pt,
	parbox=false,
	before skip=8pt,after skip=10pt,
	colbacklower=blue!6!white,
	borderline west={1mm}{0mm}{blue!30!white}}

\usepackage[sort&compress,nameinlink]{cleveref}

\title{{\huge\sffamily Proportional Participatory Budgeting \\ with Additive Utilities}}

\usepackage{authblk}

\author[1]{Dominik Peters}
\author[2]{Grzegorz Pierczy\'{n}ski}
\author[2]{Piotr Skowron}
\affil[1]{CNRS, LAMSADE, Universit\'e Paris Dauphine--PSL, \texttt{dominik.peters@lamsade.dauphine.fr}}
\affil[2]{University of Warsaw, $\{$\texttt{g.pierczynski,p.skowron}$\}$\texttt{@mimuw.edu.pl}}

\date{Version v2 -- October 2022\thanks{In arXiv version v2, we have updated the paper to follow the version at NeurIPS 2021 \citep{peters2021neurips}, though this version has significant additional improvements. Changes: The paper title now refers to ``additive'' rather than to ``cardinal'' valuations. We now call our voting rule \emph{Equal Shares} instead of the preliminary name ``Rule X''. We have added pseudocode and examples. We now use a slightly stronger definition of EJR up to one project, see \Cref{fn:ejr-counterexamples}. We have clarified the description of Equal Shares for ordinal valuations. We have added a remark that GCR satisfies FJR even for general monotone valuations. In addition, the presentation has been clarified throughout.}}
\begin{document}
	
	\maketitle

\begin{abstract}
We study voting rules for participatory budgeting, where a group of voters collectively decides which projects should be funded using a common budget. We allow the projects to have arbitrary costs, and the voters to have arbitrary additive valuations over the projects. We formulate an axiom (Extended Justified Representation, EJR) that guarantees proportional representation to groups of voters with common interests. We propose a simple and attractive voting rule called the Method of Equal Shares that satisfies this axiom for arbitrary costs and approval utilities, and that satisfies the axiom up to one project for arbitrary additive valuations. This method can be computed in polynomial time. In contrast, we show that the standard method for achieving proportionality in committee elections, Proportional Approval Voting (PAV), cannot be extended to work with arbitrary costs. Finally, we introduce a strengthened axiom (Full Justified Representation, FJR) and show that it is also satisfiable, though by a computationally more expensive and less natural voting rule. 
\end{abstract}

\section{Introduction}
A growing list of cities now uses Participatory Budgeting (PB) to decide how to spend their budgets \citep{de2021international,wampler2021participatory}. Through a voting system, PB allows the residents of a city to decide which projects will be funded by the city authorities. This increases civic involvement in government, by increasing the number of issues that are decided by democratic vote, and by allowing residents to submit their own project proposals \citep{participatoryBudgeting,aziz2020participatory}. 

To count the votes, most cities use a variant of a simple protocol: Each voter is allowed to vote for a certain number of project proposals. Then, the projects with the highest number of votes are funded, until the budget limit is reached. While simple and intuitive, this is a bad voting rule, because it gives too much voting power to pluralities: even quite small cohesive groups (if they are larger than other groups) can dictate the entire outcome, an effect that one might call \myemph{the tyranny of the largest minority}. 

To see this, consider Circleville, a fictional city divided into four districts. A map of the city is shown in \Cref{fig:circleville}. The districts all have similar sizes, but Northside has the largest population. Suppose \$400k have been allocated to PB (\$1 per person), and suppose that all the project proposals are of a local character (such as school renovations), and that residents only vote for projects that concern their own district. For example, every Northside resident will cast votes for projects $A$, $B$, $C$, and $D$, but no one else votes for these. Because Northside is the most populous district, the Northside projects will all receive the highest number of votes, and the voting rule described will spend the entire budget on Northside projects (by implementing projects $C$ and $D$). The 280k residents of the other three districts are left empty-handed.

\begin{figure}
	\centering
	\begin{tikzpicture}[transform shape, scale=0.9]
\draw [thick] (0,0) circle (5);
	\foreach \x in {45,135,225,-45}
		\draw [thick] (\x:0) -- (\x:5);
	
\foreach \x/\y/\name/\pop in {
		0/4.5/Northside/120k,
		0/-4.2/Southside/80k,
		3.95/1/Eastside/110k,
		-3.9/1/Westside/90k
	}
	{
		\node [fill=black!10!white, rounded corners=0.15cm, inner sep=0.15cm] at (\x,\y) () {\name};
		\node [font=\small] at (\x,\y-0.5) () {pop. \pop};
	}
	
	\node [font=\Huge, fill=white, inner sep=3pt] at (0,1.4) () {Circleville};
	
\draw [blue!30!white,line width=0.5cm] plot [smooth] coordinates {(200:6) (190:4.6) (205:3) (195:1) (205:0) (-30:1.6) (-10:2.8) (-10:4) (-18:6)};
	
	\node [font=\scriptsize,color=blue,rotate=-34] at (4.8,-1.25) (river) {Example River};
	
\fill [green!50!gray!80!white, rounded corners=0.2cm] (-0.3,2) -- (0,3) -- (1.5,3) -- (1.2,2);
	\fill [green!50!gray!80!white] (-4.4,-2) -- (-4.0,-2) -- (-3.8,-2.4) -- (-4.3,-2.4);
	
\foreach \x/\y/\name/\price in {
		-2.55/2.8/A/50k,
		-1.95/3.8/B/30k,
		0.2/2.7/C/150k,
		2.2/2.8/D/250k,
		3.1/2.1/E/60k,
		2.3/1.4/F/10k,
		2.18/-0.2/G/90k,
		0.95/-0.18/H/60k,
		3.0/-1.8/I/4k,
		1.0/-2.5/J/70k,
		2.0/-3.6/K/20k,
		1.3/-4/L/30k,
		0.0/-3.2/M/20k,
		-1.7/-3.5/N/40k,
		-1/-1.9/O/100k,
		-3.3/-3.0/P/30k,
		-4.16/-2.2/Q/80k,
		-3.7/-0.6/R/10k,
		-1.8/-0.3/S/40k,
		-3.7/2.0/T/7k
	}
	{
		\draw (\x,\y) -- (\x,\y+0.5);
		\draw [fill=white] (\x,\y+0.5) circle (0.1);
		\draw [fill=black] (\x,\y) circle (0.01);
		\node [anchor=west] at(\x+0.1,\y+0.65) () {$\name$};
		\node [fill=blue!10!white, rounded corners=0.15cm, inner sep=0.07cm, font=\scriptsize, anchor=west] at(\x+0.2,\y+0.25) () {\$\price};
	}
	
	\end{tikzpicture}
	\caption{Map of Circleville, showing the locations and costs of the PB project proposals.}
	\label{fig:circleville}
\end{figure}

To circumvent this obvious issue, many cities have opted to hold separate elections for each district. The budget is divided in advance between the districts (e.g., in proportion to their number of residents), each project is assigned to a district, and voters only vote in their local election. While this avoids the issue of spending the entire budget in Northside, this fix introduces many other problems. For example, projects on the boundary of two districts (such as $A$ and $P$) need to be assigned to one of them. Residents of the other district may be in favor of the boundary project, but cannot vote for it. Thus boundary projects are less likely to be funded, even if they would be more valuable overall. Similarly, projects without a specific location that benefit the entire city are difficult to handle. Also, interest groups that are not geographic in nature will be underserved; for instance, parents across the city might favor construction of a large playground (project $C$), but with separate district elections, parents cannot form a voting block.
Similarly, bike riders across the city cannot express their joint interest in the construction of a bike trail along Example River (projects $R$, $S$, $H$, and $G$).

To solve these problems, it seems desirable to hold a single city-wide election, but to use a voting system that ensures that money is spent proportionally. The voting system should automatically and endogenously identify groups of voters who share common interests, and make sure that those groups are appropriately represented. This aim has been identified by several researchers \citep{aziz2018proportionally,fain2018fair}, but in our view no convincing proposal for a proportional voting rule has emerged before this work. Good formalizations of “proportionality” for the PB context were also missing; the concept of the \emph{core} is a candidate but it is a very demanding requirement, and there are situations where it cannot be satisfied \citep{fain2018fair}.

In this paper, we formalize proportionality for participatory budgeting as an axiom called \emph{extended justified representation} (EJR). The axiom requires that no group of voters with common interests is underserved. We construct a simple and attractive voting rule (the Method of \emph{Equal Shares}) that satisfies EJR for approval preferences, and that satisfies EJR up to one project for general additive valuations. We then discuss a strengthening of EJR---which we call \emph{fully justified representation} (FJR)---and show that this strengthening is still satisfiable, albeit by a different voting rule. 

\subsubsection*{Our approach: Generalize concepts from multi-winner voting}
Both our proportionality axiom and our voting rule are generalizations of concepts that have been introduced in the literature on multi-winner voting \citep{FSST-trends}. That literature can be seen as handling a special case of PB, where all projects cost the same amount of money. This is often called the \emph{unit cost assumption}. Under this assumption, the problem is equivalent to selecting a committee of a specified size $k$. 

Much of the relevant literature studies rules that work with \emph{approval ballots} \citep{lackner2022book}, where voters are allowed to approve or disapprove each project. Approvals are interpreted as follows: a voter's utility of an outcome set $W$ of projects is the number of projects in $W$ that the voter has approved.
In the following, we will relax these two assumptions (approvals, unit costs) one after the other: first we consider generalizing approval-based rules beyond the unit cost case, and then we will move to \emph{additive valuations} which allow voters to specify utilities beyond just 0 and 1.

\subsubsection*{Proportional Approval Voting does not remain proportional beyond unit costs}

The study of approval-based multi-winner voting rules has been very productive \citep{justifiedRepresentation,aaai/BrillFJL17-phragmen,pjr17,lac-sko:t:approval-thiele,pet-sko:laminar}. Researchers have identified a considerable number of proportionality axioms and of attractive voting rules for this case. 

The most prominent proposal for a proportional multi-winner rule is known as Proportional Approval Voting (PAV), also known as Thiele's method after its Danish first inventor \citet{Thie95a}. Thiele's rule is based on optimization. Suppose that $N$ is the set of voters, and that each voter $i\in N$ has indicated a set $A(i)$ of projects that $i$ approves. Then for each set $W$ of projects which is feasible (i.e., its total cost is at most the budget limit), the rule computes the score
\[ \text{PAV-score}(W) = \sum_{i\in N} \left(1 + \frac12 + \frac13 + \cdots + \frac{1}{|W \cap A(i)|} \right) \text{.} \]
The output of PAV is a feasible set $W$ that maximizes this score. When the unit cost assumption holds, PAV is a great rule and lives up to its name: it is known to be proportional both in an axiomatic sense \citep{justifiedRepresentation} and in a quantitative sense \citep{skowron:prop-degree}. In fact, it is known that among ``optimization-based rules'' (suitably defined), only PAV is proportional \citep{lac-sko:t:approval-thiele,justifiedRepresentation}.

\begin{figure}
	\centering
	\begin{tikzpicture}
	\draw [thick, rounded corners=0.5cm] (0,0) rectangle (6.66,3);
	\draw [thick] (4.44,0) -- (4.44,3);
	\node [font=\LARGE, fill=white, inner sep=5pt] at (3.33,3) () {Onetown};
	
	\foreach \x/\y/\name/\pop in {
		2.3/2.2/Leftside/60k,
		5.55/2.2/Rightside/30k
	}
	{
		\node [fill=black!10!white, rounded corners=0.15cm, inner sep=0.15cm] at (\x,\y) () {\name};
		\node [font=\small] at (\x,\y-0.5) () {pop. \pop};
	}
	
	\foreach \x/\y/\name/\price in {
		0.4/0.5/L_1/20k,
		1.8/0.5/L_2/20k,
		3.2/0.5/L_3/20k,
		5.1/0.5/R/45k
	}
	{
		\draw (\x,\y) -- (\x,\y+0.5);
		\draw [fill=white] (\x,\y+0.5) circle (0.1);
		\draw [fill=black] (\x,\y) circle (0.01);
		\node [anchor=west] at(\x+0.1,\y+0.65) () {$\name$};
		\node [fill=blue!10!white, rounded corners=0.15cm, inner sep=0.07cm, font=\scriptsize, anchor=west] at(\x+0.2,\y+0.25) () {\$\price};
	}
	\end{tikzpicture}
	\qquad
	\begin{tikzpicture}
	\draw [thick, rounded corners=0.5cm] (0,0) rectangle (6.66,3);
	\draw [thick] (4.44,0) -- (4.44,3);
	\node [font=\LARGE, fill=white, inner sep=5pt] at (3.33,3) () {Twotown};
	
	\foreach \x/\y/\name/\pop in {
		2.3/2.2/Leftside/60k,
		5.55/2.2/Rightside/30k
	}
	{
		\node [fill=black!10!white, rounded corners=0.15cm, inner sep=0.15cm] at (\x,\y) () {\name};
		\node [font=\small] at (\x,\y-0.5) () {pop. \pop};
	}
	
	\foreach \x/\y/\name/\price in {
		0.4/0.5/L_1/30k,
		1.8/0.5/L_2/30k,
		3.2/0.5/L_3/30k,
		5.1/0.5/R/30k
	}
	{
		\draw (\x,\y) -- (\x,\y+0.5);
		\draw [fill=white] (\x,\y+0.5) circle (0.1);
		\draw [fill=black] (\x,\y) circle (0.01);
		\node [anchor=west] at(\x+0.1,\y+0.65) () {$\name$};
		\node [fill=blue!10!white, rounded corners=0.15cm, inner sep=0.07cm, font=\scriptsize, anchor=west] at(\x+0.2,\y+0.25) () {\$\price};
	}
	\end{tikzpicture}
	\caption{Onetown and Twotown are identical, except that the projects have different costs. Both have a budget of \$90k available for PB.}
	\label{fig:equivalent-instances}
\end{figure}

However, when the unit cost assumption does not hold, PAV ceases to guarantee proportional representation. To see this, consider the city of Onetown shown in \Cref{fig:equivalent-instances}. Onetown has 90,000 residents split in two districts, and has \$90,000 available for participatory budgeting. The 60,000 residents in Leftside all vote for projects $\{L_1, L_2, L_3\}$ and the 30,000 residents in Rightside vote for the single project $\{R\}$. Note that the costs are such that we can either afford to implement all three $L$-projects, giving PAV score 110,000, or implement two $L$-projects and the $R$-project, giving PAV score 120,000. Thus, PAV implements two $L$-projects and the $R$-project. However, note that Leftside residents form two thirds of the population of Onetown, and so by proportionality are entitled to two thirds of the budget (\$60,000), which is enough to implement all three $L$-projects. Hence, Leftside is underrepresented by PAV.

To see what is going on, consider Twotown from \Cref{fig:equivalent-instances}. Twotown is just like Onetown, except that now all projects cost \$30,000. Note that for Twotown it is still the case that we can either afford all three $L$-projects, or two $L$-projects plus $R$. By the same calculation as before, PAV implements the latter possibility. This time, this is the proportional choice: Leftside now deserves only two projects, since only two projects are affordable with Leftside's share of the budget.

Onetown and Twotown are nearly identical: same number of residents, same district structure, same alternatives, same approval sets, and even the feasibility constraint (three $L$s or two $L$s plus $R$) is the same. Since the definition of PAV only depends on these characteristics, it must select the same outcome for both towns. But the costs differ, and therefore different outcomes are proportional, and hence PAV fails proportionality. The same is true for all other rules that depend only on preferences and feasibility constraints but not costs. This suggests that there is no variant of  the PAV idea of optimzing some social welfare function that retains PAV's proportionality guarantees beyond the unit cost case.

\begin{theorem}
	\label{thm:welfarist-impossibility}
	Every voting rule that only depends on voters' utility functions and the collection of budget-feasible sets must fail proportionality, even on instances with a district structure.
\end{theorem}

In fact, as we will show later (\Cref{ex:pav_no_ejr}), PAV is not even approximately proportional.

\subsubsection*{Equal Shares: A simple method that guarantees proportionality}

Besides Thiele's method, another committee method that satisfies strong proportionality properties is the recently proposed \emph{Method of Equal Shares} \citep{pet-sko:laminar}, originally known as ``Rule X''.  It turns out that this rule can be naturally extended beyond unit costs, and that (unlike PAV) the rule retains its proportionality properties in the general case.

On a high level, Equal Shares works as follows. If there are $n$ voters, it starts out by dividing the available budget into $n$ parts, and giving each voter their share into a virtual bank account. Equal Shares then repeatedly looks for a project whose approvers have enough money left to fund it, chooses one such project and charges its approvers for it. The rule does so until no further projects are affordable. Just from this description it is clear the rule will be proportional on instances with a district structure, because if the member of a district only approve projects in their district, then the rule will spend (at least) a fraction of the budget proportional to the district size on those projects. But to ensure good behavior on other instance, it is crucial exactly how the rule chooses among different affordable projects and how the rule divides the chosen project's cost among its supporters. Equal Shares always spreads the cost of the project as \emph{equally as possible}, which usually means that all supporters contribute the same amount of money to it; if some supporters do not have this amount of money left, they spend their entire remaining budget. If several projects are affordable, Equal Shares chooses the project that minimzes the highest amount that any supporter needs to pay. (Thus, all else equal, Equal Shares favors cheap projects over expensive ones, and favors projects with many supporters over projects with fewer.)

As we mentioned, it is clear that Equal Shares is proportional on district-based instances. On its own, this is a rather weak guarantee. In the real world, like in Circleville (\Cref{fig:circleville}), voters will sometimes vote for projects in other districts, and it is uncommon that voters will approve \emph{all} the projects in their own district. A truly proportional voting rule should be able to represent all kinds of interest groups, even in cases where the same voter is part of several such groups.

Consider an arbitrary subset of voters, $S \subseteq N$. For example, $S$ could be the residents of a district, or the set of parents in Circleville, or the set of bike users. The group $S$ forms a fraction $|S|/|N|$ of the population, and thus intuitively its members deserve to control a fraction $|S|/|N|$ of the budget. This idea is the basis of most proportionality axioms developed in the literature; they differ by how they formalize the notion of ``deserving'' part of the budget. We will consider an axiom that guarantees to represent groups whenver they are sufficiently \emph{cohesive}, in the sense of having similar preferences. Suppose that $S$ can come up with a set $T$ of projects such that $T$ can be funded with a $|S|/|N|$ fraction of the budget. Suppose further that each voter in $S$ approves all the projects in $T$; this means the group is cohesive. Then an axiom called Extended Justified Representation (EJR) demands that the voting rule must select a set $W$ such that at least one voter in $S$ approves at least $|T|$ of the funded projects in $W$. In other words, EJR prohibits sets $W$ where all the voters in $S$ are underrepresented in the sense that they would all prefer the set $T$ to $W$.

EJR was first proposed in the context of committee voting by \citet{justifiedRepresentation}. EJR is a demanding property, and initially PAV was the only known natural voting rule satisfying it, but we have seen that without unit costs, PAV fails EJR even in well-structured cases. 
However, \citet{pet-sko:laminar} showed that Equal Shares satisfies EJR for approval-based committee elections. In this paper, we will prove that this continues to hold without the unit cost assumption.

%\begin{theorem}
%	For approval-based PB, Equal Shares satisfies Extended Justified Representation.
%\end{theorem}
%\addtocounter{theorem}{-1}
\medskip
\noindent
\textbf{Theorem 2.} \emph{For approval-based PB, Equal Shares satisfies Extended Justified Representation.}
\medskip

The intuition behind this result is that, under Equal Shares, a group $S$ is explicitly given their share of the budget. As the rule progresses, the money of $S$ is spent and by design of Equal Shares it is spent on projects that provide good value for money. Thus, the only way that $S$ could end up underrepresented is if Equal Shares does not spend enough of $S$'s money; but we can show that this never happens if $S$ is cohesive.

\subsubsection*{Beyond approval: Allowing more expressive additive preferences}

In real-world PB elections, different projects differ vastly in their scopes and costs. For example, in the 2019 PB election in the 16th arrondissement of Paris, the most expensive project that was funded cost \EUR{560k} (refurbishing a sports facility) and the cheapest cost \EUR{3k} (providing materials for a school project of building a board game).
The former project received 775 votes, and the latter 670 votes. Hence, the former project was 1.15 times as popular as the latter, but it cost 186 times as much! If we interpret approvals as 0/1 utilities, we thereby assume that a voter is indifferent between them. On that view, the cheap project provides an amazing value. It is more likely, though, that the approval-based interface did not allow voters to adequately express their values.

Facing these large cost differences, a better preference model might be given by general \emph{additive valuations}, which allow voters to specify an arbitrary utility value for each project, with the assumption that a voter's satisfaction is the sum of the utilities of the funded projects. Some cities already allow residents to vote this way; for example the Polish city of Cz\c{e}stochowa asks voters to divide 10 points between projects (for example they could assign 1 point each to 10 different projects, or give 6 points to one project and 4 points to another), while Gdansk and Strasbourg ask voters to divide 5 points between projects. In the literature, additive valuations are considered by \citet{benade2017preference} who study preference elicitation issues, and by \citet{fain2018fair} and \citet{fluschnik2019fair} who consider an aggregation rule similar to PAV, based on optimizing a Nash product objective. The latter rule will not satisfy us, given PAV's poor performance without unit costs. Further, even for unit costs, the rule does not satisfy our version of the EJR axiom.

We propose a way to adapt Equal Shares to general additive valuations. In our proposal, when Equal Shares decides to fund a project, a voter's payment is proportional to the voter's utility for the project. So if voter $i$ assigns utility 2 to project $B$ while $j$ assigns utility 1 to $B$, then Equal Shares will ask $i$ to pay twice as much as $j$ if $B$ is funded. We also propose a natural way to extend the EJR axiom to general additive valuations. Equal Shares almost satisfies the axiom: it satisfies Extended Justified Representation up to one project, which means that for every group, it would be appropriately represented if we are allowed to add one more project to the outcome. We do show that there always exists an EJR outcome (without the ``up to one'' qualification), but also prove that unless P = NP, no strongly polynomial-time method such as Equal Shares can satisfy EJR.

While eliciting additive valuations seems attractive to obtain fine-grained utility information, many cities may choose to stick to approval voting for its simplicity. It is still valuable to have generalized Equal Shares to additive valuations, because it allows us to interpret approval ballots in a different way. So far, we have assigned utility 1 to approved projects $c$, but we could also assign them utility $\cost(c)$, thus obtaining \emph{cost utilities}. The latter way assumes that a voter approving two projects with different costs will prefer funding the more expensive one. In the Paris example above, this seems closer to the truth. In other words, for approval utilities we implicitly assume that the voter's utility is defined as the \myemph{number of projects} that the voters supports and that have been funded; for cost utilities, the voter's utility is defined as the total \myemph{amount of funds spent on projects} supported by the voter.

Finally, additive valuations allow us to consider another input format, where voters \emph{rank} the projects in order of preference (\emph{ordinal preferences}). We show that if we convert rankings into additive valuations using a lexicographic scheme, then Equal Shares turns into a voting method that satisfies a property known as Proportionality for Solid Coalitions (PSC), an axiom identified in the literature on the Single Transferable Vote (STV) \citep{woo:j:properties,tid-ric:j:stv,az-bar:expanding_approval} and that is desirable for electing committees given ordinal preferences.

\subsubsection*{FJR: A proportionality axiom even stronger than EJR}

In the previous literature on approval-based multiwinner elections, loosely speaking, EJR was the strongest proportionality axiom known to always be satisfiable.\footnote{Though the literature contains other proportionality notions that are both logically and conceptually incomparable, such as ``perfect representation'' \citep{pjr17} satisfied by the Monroe rule, and the concepts of laminar proportionality and priceability \citep{pet-sko:laminar} satisfied by Phragm\'en's rule and Equal Shares.}
Many other rules such as Phragm\'en's rule or Chamberlin--Courant only satisfy substantially weaker axioms (known as PJR \citep{pjr17} and JR \citep{justifiedRepresentation}).

A very attractive strengthening of EJR is the \emph{core} \citep{justifiedRepresentation,fain2018fair}. We say that a set $S \subseteq N$ of voters \emph{blocks} an outcome $W$ if there is a set $T$ of  projects affordable with a $|S|/|N|$ fraction of the budget such that each member of $S$ strictly prefers $T$ to $W$ (in the sense that each member of $S$ approves strictly more projects in $T$ than in $W$). In such a case, the group $S$ appears to be underrepresented. An outcome $W$ is in the core if it is not blocked by any coalition $S$. (EJR is a special case of the core where the set $T$ needs to be unanimously liked by the members of $S$, so that $S$ is a ``cohesive'' group.) For the approval-based case, it is unknown whether there always exists an outcome in the core (even under the unit cost assumption), and this is surely the most tantalizing open problem in this area. For general additive utilities, we already know that the core might be empty.

We propose a property that is weaker than the core but stronger than EJR.
We call this axiom Fully Justified Representation (FJR).\footnote{Apologies that this name is not particularly descriptive, but then neither is EJR or PJR.}
It strengthens EJR by weakening the cohesiveness requirement.
FJR requires that if a group $S\subseteq N$ of voters can propose a set $T$ of projects that is affordable with $S$'s share of the budget, and each voter has utility at least $\ell$ for the set $T$, then at least one voter in $S$ has utility at least $\ell$ in the chosen outcome $W$. In the approval case, we see that FJR does not insist that $T$ is unanimously approved by the group $S$ (like in EJR), but it just requires that $T$ is very popular among $S$.

Both PAV and Equal Shares fail FJR (\Cref{ex:rulex_no_strong_ejr,ex:pav_no_strong_ejr}). However, we prove that there does indeed exist a rule satisfying FJR, which works for arbitrary costs. The rule is a simple greedy procedure that repeatedly looks for groups with maximum cohesiveness and then satisfies them. While this is not a polynomial-time algorithm and not a particularly natural rule, we can show that this proposal can be made compatible with some other properties (in particular priceability, a property introduced by \citet{pet-sko:laminar}). For future work, it will be interesting to look for new natural rules satisfying FJR; this is even interesting for the committee context. An attractive feature of FJR and of our greedy rule for satisfying it is that they work for general monotone valuations, and so we do not even need to assume additivity.

\section{Preliminaries}\label{sec:preliminaries}

For each $t \in \naturals$, write $[t] = \{1,2,\ldots,t\}$.
An \myemph{election} is a tuple $(N, C, b, \cost, \{u_i\}_{i \in N})$, where:
\begin{itemize}
	\item $N=[n]$ and $C = \{c_1, \ldots, c_m\}$ are the sets of \myemph{voters} and  \myemph{candidates} (or \myemph{projects}).
	\item $b \in \rationals_{\ge 0}$ is the available \myemph{budget}.
	\item $\cost\colon C \to \rationals_{\ge 0}$ is a function that for each $c \in C$ assigns the \myemph{cost} that needs to be paid if $c$ is selected. For each $T \subseteq C$, we write $\cost(T) = \sum_{c \in T} \cost(c)$ for the total cost of $T$.
	\item For each voter $i \in N$, the function $u_i\colon C \to \rationals_{\ge 0}$ defines $i$'s additive \myemph{utility function}.\footnote{The rules we propose will be invariant under rescaling all utility functions by a common factor, so without loss of generality we could assume that $u_i(c) \in [0,1]$ for all $i \in N$ and $c \in C$. The rules will not be invariant under rescaling by different factors, or by shifting utilities.}
	If a set $T \subseteq C$ of candidates is implemented, $i$'s overall utility is $u_i(T) =  \sum_{c \in T} u_i(c)$. For a subset $S \subseteq N$ of voters, we further write $u_S(T) = \sum_{i \in S} \sum_{c \in T} u_i(c)$ for the total utility enjoyed by $S$ if $T$ is implemented.
	We assume that $u_N(c) > 0$ for each $c \in C$, so that every candidate is assigned positive utility by at least one voter.
\end{itemize}

A subset of candidates $W \subseteq C$ is \myemph{feasible} if $\cost(W) \leq b$. Our goal is to choose a feasible subset of candidates, which we call an \myemph{outcome}, based on voters' utilities. An \myemph{aggregation rule} (or, in short, a \myemph{rule}) is a function $\mathcal{R}$ that for each election $E$ selects a feasible outcome $W = \mathcal{R}(E)$ called the \myemph{winning outcome}.\footnote{Sometimes there are ties. For the results of this paper it does not matter how these ties are broken.} We call this the \myemph{general PB model}.

There are two interesting special cases of our model:
\begin{description}
	\item[Committee elections.] In this case, the budget is an integer, $b \in \naturals$ (the committee size) and that each candidate costs $1$. Then $W$ is an outcome if and only if  $|W| \leq b$. In this special case we also refer to outcomes as \emph{committees}, and we say that the election satisfies the \myemph{unit cost assumption}.
	\item[Approval utilities.] In this case, for each $i \in N$ and $c \in C$ it holds that $u_i(c) \in \{0,1\}$. The \emph{approval set} of voter $i$ is $A(i):=\{c\in C\colon u_i(c) = 1\}$, and we say that $i$ \emph{approves} candidate~$c$ if $c \in A(i)$. If $c \in A(i) \cap W$, we say that $c$ is a \emph{representative} of $i$ in $W$.
\end{description}
Often we combine of these special cases, and study \myemph{approval-based committee elections}.

\section{The Method of Equal Shares}\label{sec:rulex}

Recently, \citet{pet-sko:laminar} introduced an aggregation rule  for approval-based committee elections that they called Rule X. In that setting (approval-based committee elections) the rule satisfies a combination of appealing proportionality properties. Here, we extend it to the more general model of participatory budgeting, that is, to the model with arbitrary costs and utilities.  We will call this rule the \emph{Method of Equal Shares}, because it works by dividing the available budget equally between voters, and (to the extent it is possible) sharing the cost of each project equally between the voters who approve the project. For brevity, we often refer to the rule as simply \emph{Equal Shares}.

\begin{definition}[Method of Equal Shares]\label{def:rulex}
	Each voter is initially given an equal fraction of the budget, i.e., each voter is given $\nicefrac{b}{n}$ monetary units. 
	We start with an empty outcome $W = \emptyset$ and sequentially add candidates to $W$. To add a candidate~$c$ to~$W$, we need the voters to pay for $c$. For each selected candidate $c \in W$ we write $p_i(c)$ for the amount that voter $i$ pays for $c$; we need that $\sum_{i \in N} p_i(c) = \cost(c)$. We write $p_i(W) = \sum_{c \in W} p_i(c) \le \nicefrac{b}{n}$ for the total amount $i$ has paid so far; then write $b_i = \nicefrac{b}{n} - p_i(W)$ for the amount of money that $i$ has left (clearly, in the first round $b_i = \nicefrac{b}{n}$ for all $i \in N$). For $\rho \ge 0$, we say that a candidate $c \not\in W$ is \emph{$\rho$-affordable} if
	\begin{align*}
		\sum_{i \in N} \min\left( b_i , u_i(c) \cdot \rho \right) = \cost(c) \text{.}
	\end{align*}
	If no candidate is $\rho$-affordable for any $\rho$, Equal Shares terminates and returns $W$. Otherwise it selects a candidate $c\not\in W$ that is $\rho$-affordable for a minimum $\rho$. 
	Individual payments are given by
	\[
	\pushQED{\qed} 
	p_i(c) = \min\left( b_i, u_i(c)\cdot \rho \right) 
	\qedhere
	\popQED
	\]
\end{definition}

Intuitively, when Equal Shares adds a candidate $c$, it asks voters to pay an amount proportional to their utility $u_i(c)$ for $c$; in particular, the cost per unit of utility is $\rho$. If a voter does not have enough money, the rule asks the voter to pay all the money the voter has left, which is $b_i = \nicefrac{b}{n} - p_i(W)$. Throughout the execution of Equal Shares, the value of $\rho$ increases. Thus, candidates are added in decreasing order of utility per dollar that the voters get from the candidates.\footnote{An exception are voters who do not have enough money to contribute $u_i(c) \cdot \rho$ and instead contribute all of their remaining money (which could be nothing); those voters obviously get higher utility per dollar.}

\subsubsection*{Special Case: Approval Utilities}

In the special case where voters have approval utilities ($u_i(c) \in \{0,1\}$ for all $i \in N$ and $c \in C$), the definition of Equal Shares becomes somewhat more straightforward. At each step, a candidate $c \not\in W$ is $\rho$-affordable if and only if the cost of $c$ can be covered by the voters approving $c$ in such a way that the maximum payment of any voter is $\rho$. Voters who have less than $\rho$ left spend all their money, and the other voters pay exactly $\rho$. This way, the cost is shared as equally as possible among supporters of a project, which motivates the name Method of Equal Shares. See the following example for an illustration of the determination of $\rho$.

\subsubsection*{Example}

Let us go through a small example with 5 projects, 10 voters $N = \{i_1, \dots, i_{10}\}$ with approval utilities, and a budget of $b = \$100$. The project costs and utilities are shown in the table below.

\begin{center}
\newcommand{\tblapprove}{\tikz[baseline=0pt]{\node [font=\small,minimum width=13pt, fill=blue!20!white,anchor=base] {$1$}}}
\newcommand{\tblnotapprove}{\tikz[baseline=0pt]{\node [font=\small,minimum width=13pt, fill=black!5,anchor=base] {$0$}}}
\newcommand{\tblcost}[1]{\tikz[baseline=0pt]{\node [font=\small,minimum width=2pt*#1, fill=black!15,inner sep=2pt,anchor=base,text width=2.7pt*#1-7pt] {\$#1}}}

\begin{tabular}{llcccccccccc}
	\toprule
	& $\cost$ & $i_1$ & $i_2$& $i_3$& $i_4$& $i_5$& $i_6$& $i_7$& $i_8$& $i_9$& $i_{10}$ \\
	\midrule
	Project 1 & \tblcost{36} & \tblapprove & \tblapprove & \tblapprove& \tblapprove& \tblapprove& \tblnotapprove& \tblnotapprove& \tblnotapprove& \tblnotapprove& \tblnotapprove \\[1pt]
	Project 2 & \tblcost{36} & \tblapprove & \tblapprove & \tblapprove& \tblapprove& \tblapprove& \tblapprove& \tblnotapprove& \tblnotapprove& \tblnotapprove& \tblnotapprove \\[1pt]
	Project 3 & \tblcost{25} & \tblapprove & \tblapprove & \tblapprove& \tblapprove& \tblapprove& \tblnotapprove& \tblnotapprove& \tblnotapprove& \tblnotapprove& \tblnotapprove \\[1pt]
	Project 4 & \tblcost{24} & \tblnotapprove & \tblnotapprove & \tblnotapprove& \tblapprove& \tblapprove& \tblapprove& \tblapprove& \tblnotapprove& \tblnotapprove& \tblnotapprove \\[1pt]
	Project 5 & \tblcost{24} & \tblnotapprove & \tblnotapprove & \tblnotapprove& \tblnotapprove& \tblnotapprove& \tblapprove& \tblapprove& \tblapprove& \tblapprove& \tblnotapprove \\
	\bottomrule
\end{tabular}
\end{center}
Sharing the available budget equally among the voters, everyone starts with \$10. We can now calculate, for each project, the value $\rho$ such that the project is $\rho$-affordable. As discussed above, because we have approval utilities, this just entails spreading the cost as equally as possible among project supporters. 
\begin{center}
\begin{minipage}{0.4\linewidth}
\begin{tikzpicture}[xscale=0.5,yscale=0.7]
	\newcommand{\budgetscale}{0.2}
	\foreach 
	\x/\approving/\budget/\payment
	in
	{
		1/1/10/7.2,
		2/1/10/7.2,
		3/1/10/7.2,
		4/1/10/7.2,
		5/1/10/7.2,
		6/0/10,
		7/0/10,
		8/0/10,
		9/0/10,
		10/0/10
	}
	{
		\if\approving1
		\fill[blue!30!white] (\x-0.4,0) rectangle (\x+0.4,\budgetscale*\budget);
		\node [anchor=south, font=\footnotesize] at (\x, \budgetscale*\budget) {\budget};
		\fill[blue!50!black!70!white] (\x-0.4,0) rectangle (\x+0.4,\budgetscale*\payment);
		\node [anchor=south, text=white, font=\footnotesize] at (\x, 0) {\payment};
		\else
		\fill[black!20] (\x-0.4,0) rectangle (\x+0.4,\budgetscale*\budget);
		\node [anchor=south, font=\footnotesize] at (\x, \budgetscale*\budget) {\budget};
		\fi
		\node [font=\footnotesize] at (\x, -0.3) {$i_{\x}$};
	}
\def\rhovalue{7.2}
	\draw [densely dashed, thick] (0.5, \budgetscale*\rhovalue) -- (11, \budgetscale*\rhovalue);
	\node [anchor=west,font=\footnotesize] at (10.95, \budgetscale*\rhovalue - 0.05) {$\rho = \rhovalue$};
\draw (0.3,-0.65) rectangle (10.7,\budgetscale*10+0.85);
	\node[font=\footnotesize, fill=white] at (5, \budgetscale*10+0.8) {Project 1};
\end{tikzpicture}
\end{minipage}
\quad
\begin{minipage}{0.4\linewidth}
\begin{tikzpicture}[xscale=0.5,yscale=0.7]
	\newcommand{\budgetscale}{0.2}
	\foreach 
	\x/\approving/\budget/\payment
	in
	{
		1/1/10/6,
		2/1/10/6,
		3/1/10/6,
		4/1/10/6,
		5/1/10/6,
		6/1/10/6,
		7/0/10,
		8/0/10,
		9/0/10,
		10/0/10
	}
	{
		\if\approving1
		\fill[blue!30!white] (\x-0.4,0) rectangle (\x+0.4,\budgetscale*\budget);
		\node [anchor=south, font=\footnotesize] at (\x, \budgetscale*\budget) {\budget};
		\fill[blue!50!black!70!white] (\x-0.4,0) rectangle (\x+0.4,\budgetscale*\payment);
		\node [anchor=south, text=white, font=\footnotesize] at (\x, 0) {\payment};
		\else
		\fill[black!20] (\x-0.4,0) rectangle (\x+0.4,\budgetscale*\budget);
		\node [anchor=south, font=\footnotesize] at (\x, \budgetscale*\budget) {\budget};
		\fi
		\node [font=\footnotesize] at (\x, -0.3) {$i_{\x}$};
	}
\def\rhovalue{6}
	\draw [densely dashed, thick] (0.5, \budgetscale*\rhovalue) -- (11, \budgetscale*\rhovalue);
	\node [anchor=west,font=\footnotesize] at (10.95, \budgetscale*\rhovalue - 0.05) {$\rho = \rhovalue$};
\draw (0.3,-0.65) rectangle (10.7,\budgetscale*10+0.85);
	\node[font=\footnotesize, fill=white] at (5, \budgetscale*10+0.8) {Project 2};
\end{tikzpicture}
\end{minipage}
\vspace{5pt}

\noindent
\begin{minipage}{0.4\linewidth}
\begin{tikzpicture}[xscale=0.5,yscale=0.7]
	\newcommand{\budgetscale}{0.2}
	\foreach 
	\x/\approving/\budget/\payment
	in
	{
		1/1/10/5,
		2/1/10/5,
		3/1/10/5,
		4/1/10/5,
		5/1/10/5,
		6/0/10,
		7/0/10,
		8/0/10,
		9/0/10,
		10/0/10
	}
	{
		\if\approving1
		\fill[blue!30!white] (\x-0.4,0) rectangle (\x+0.4,\budgetscale*\budget);
		\node [anchor=south, font=\footnotesize] at (\x, \budgetscale*\budget) {\budget};
		\fill[blue!50!black!70!white] (\x-0.4,0) rectangle (\x+0.4,\budgetscale*\payment);
		\node [anchor=south, text=white, font=\footnotesize] at (\x, 0) {\payment};
		\else
		\fill[black!20] (\x-0.4,0) rectangle (\x+0.4,\budgetscale*\budget);
		\node [anchor=south, font=\footnotesize] at (\x, \budgetscale*\budget) {\budget};
		\fi
		\node [font=\footnotesize] at (\x, -0.3) {$i_{\x}$};
	}
\def\rhovalue{5}
	\draw [densely dashed, thick] (0.5, \budgetscale*\rhovalue) -- (11, \budgetscale*\rhovalue);
	\node [anchor=west,font=\footnotesize] at (10.95, \budgetscale*\rhovalue - 0.05) {$\rho = \rhovalue$};
\draw (0.3,-0.65) rectangle (10.7,\budgetscale*10+0.85);
	\node[font=\footnotesize, fill=white] at (5, \budgetscale*10+0.8) {Project 3};
\end{tikzpicture}
\end{minipage}
\quad
\begin{minipage}{0.4\linewidth}
\begin{tikzpicture}[xscale=0.5,yscale=0.7]
	\newcommand{\budgetscale}{0.2}
	\foreach 
	\x/\approving/\budget/\payment
	in
	{
		1/0/10,
		2/0/10,
		3/0/10,
		4/1/10/6,
		5/1/10/6,
		6/1/10/6,
		7/1/10/6,
		8/0/10,
		9/0/10,
		10/0/10
	}
	{
		\if\approving1
		\fill[blue!30!white] (\x-0.4,0) rectangle (\x+0.4,\budgetscale*\budget);
		\node [anchor=south, font=\footnotesize] at (\x, \budgetscale*\budget) {\budget};
		\fill[blue!50!black!70!white] (\x-0.4,0) rectangle (\x+0.4,\budgetscale*\payment);
		\node [anchor=south, text=white, font=\footnotesize] at (\x, 0) {\payment};
		\else
		\fill[black!20] (\x-0.4,0) rectangle (\x+0.4,\budgetscale*\budget);
		\node [anchor=south, font=\footnotesize] at (\x, \budgetscale*\budget) {\budget};
		\fi
		\node [font=\footnotesize] at (\x, -0.3) {$i_{\x}$};
	}
\def\rhovalue{6}
	\draw [densely dashed, thick] (0.5, \budgetscale*\rhovalue) -- (11, \budgetscale*\rhovalue);
	\node [anchor=west,font=\footnotesize] at (10.95, \budgetscale*\rhovalue - 0.05) {$\rho = \rhovalue$};
\draw (0.3,-0.65) rectangle (10.7,\budgetscale*10+0.85);
	\node[font=\footnotesize, fill=white] at (5, \budgetscale*10+0.8) {Project 4};
\end{tikzpicture}
\end{minipage}
\end{center}
We do not include a picture for Project 5 because it is similar to Project 4 (same cost, same number of supporters, and thus also $\rho = 6$). Equal Shares will choose to implement Project 3, since it has the lowest value of~$\rho$. Comparing the projects at this stage, we can identify a few principles:
\begin{itemize}
	\item \emph{More supporters is better}: Project 2 costs the same as Project 1, but it has more supporters. Thus its cost can be spread more thinly (i.e., a lower $\rho$), which will lead Equal Shares to prefer Project 2 over Project 1.
	\item \emph{Cheaper is better}: Project 3 has the same supporters as Project 1, but it costs less. Again this leads to a lower $\rho$, and Equal Shares will prefer Project 3 over Project 1.
\end{itemize}
In fact, the first project selected by Equal Shares will be one that (i) can be afforded by its supporters and (ii) subject to this, it has the highest number of approvers divided by cost.

After Equal Shares selects Project 3, voters $i_1$ through $i_5$ each pay \$5 for it. We can note that Projects 1 and 2 (both costing \$36) are now not affordable anymore: the supporters of Project 1 have \$25 left, and the supporters of Project 2 have \$35 left. However, Projects 4 and 5 are still affordable.
\begin{center}
\begin{minipage}{0.4\linewidth}
\begin{tikzpicture}[xscale=0.5,yscale=0.7]
	\newcommand{\budgetscale}{0.2}
	\foreach 
	\x/\approving/\budget/\payment
	in
	{
		1/0/5/6,
		2/0/5/6,
		3/0/5/6,
		4/1/5/5,
		5/1/5/5,
		6/1/10/7,
		7/1/10/7,
		8/0/10,
		9/0/10,
		10/0/10
	}
	{
		\if\approving1
		\fill[blue!30!white] (\x-0.4,0) rectangle (\x+0.4,\budgetscale*\budget);
		\node [anchor=south, font=\footnotesize] at (\x, \budgetscale*10) {\budget};
		\fill[blue!50!black!70!white] (\x-0.4,0) rectangle (\x+0.4,\budgetscale*\payment);
		\node [anchor=south, text=white, font=\footnotesize] at (\x, 0) {\payment};
		\else
		\fill[black!20] (\x-0.4,0) rectangle (\x+0.4,\budgetscale*\budget);
		\node [anchor=south, font=\footnotesize] at (\x, \budgetscale*10) {\budget};
		\fi
		\node [font=\footnotesize] at (\x, -0.3) {$i_{\x}$};
	}
\def\rhovalue{7}
	\draw [densely dashed, thick] (0.5, \budgetscale*\rhovalue) -- (11, \budgetscale*\rhovalue);
	\node [anchor=west,font=\footnotesize] at (10.95, \budgetscale*\rhovalue - 0.05) {$\rho = \rhovalue$};
\draw (0.3,-0.65) rectangle (10.7,\budgetscale*10+0.85);
	\node[font=\footnotesize, fill=white] at (5, \budgetscale*10+0.8) {Project 4};
\end{tikzpicture}
\end{minipage}
\quad
\begin{minipage}{0.4\linewidth}
\begin{tikzpicture}[xscale=0.5,yscale=0.7]
	\newcommand{\budgetscale}{0.2}
	\foreach 
	\x/\approving/\budget/\payment
	in
	{
		1/0/5/6,
		2/0/5/6,
		3/0/5/6,
		4/0/5/5,
		5/0/5/5,
		6/1/10/6,
		7/1/10/6,
		8/1/10/6,
		9/1/10/6,
		10/0/10
	}
	{
		\if\approving1
		\fill[blue!30!white] (\x-0.4,0) rectangle (\x+0.4,\budgetscale*\budget);
		\node [anchor=south, font=\footnotesize] at (\x, \budgetscale*10) {\budget};
		\fill[blue!50!black!70!white] (\x-0.4,0) rectangle (\x+0.4,\budgetscale*\payment);
		\node [anchor=south, text=white, font=\footnotesize] at (\x, 0) {\payment};
		\else
		\fill[black!20] (\x-0.4,0) rectangle (\x+0.4,\budgetscale*\budget);
		\node [anchor=south, font=\footnotesize] at (\x, \budgetscale*10) {\budget};
		\fi
		\node [font=\footnotesize] at (\x, -0.3) {$i_{\x}$};
	}
\def\rhovalue{6}
	\draw [densely dashed, thick] (0.5, \budgetscale*\rhovalue) -- (11, \budgetscale*\rhovalue);
	\node [anchor=west,font=\footnotesize] at (10.95, \budgetscale*\rhovalue - 0.05) {$\rho = \rhovalue$};
\draw (0.3,-0.65) rectangle (10.7,\budgetscale*10+0.85);
	\node[font=\footnotesize, fill=white] at (5, \budgetscale*10+0.8) {Project 5};
\end{tikzpicture}
\end{minipage}
\end{center}
Even though Projects 4 and 5 have the same number of supporters, we can see that they now induce different $\rho$-values. This points to another principle:
\begin{itemize}
	\item \emph{Richer supporters is better}: Among otherwise identical projects, those whose supporters have more money left are preferred by Equal Shares. This makes sense intuitively, because voters with more money have not yet had their preferences satisfied as much in prior rounds.
\end{itemize}
Thus, Equal Shares next selects Project 5, and voters $i_6$ through $i_9$ each pay \$6 for it. We can check that at this point, no projects are affordable (same argument as before for Projects 1 and 2, and the supporters of Project 4 only have \$18 left). Thus, Equal Shares terminates and selects the winning outcome $W = \{\text{Project 3}, \text{Project 5}\}$. Note that $\cost(W) = \$49$ while $b = \$100$, and thus a large part of the budget was not spent by Equal Shares. Indeed, we could add any one of the remaining three projects to the outcome and still satisfy the budget constraints. For strategies that ``complete'' the result of Equal Shares, see \Cref{sec:exhaustiveness}.

\subsubsection*{Implementation}
\newcommand{\money}{\ensuremath{b}}

\Cref{alg:pseudocode} shows an implementation of Equal Shares in pseudocode. The only non-obvious part of the computation is how to determine the value of $\rho$ for each candidate. Suppose that the algorithm has selected the set $W$ thus far. To calculate the value of $\rho$ for a candidate $c \in C \setminus W$,
note first that only voters with positive utility for $c$ (its ``supporters'') will pay for it. Label the supporters of $c$ as $i_1, \dots, i_t$. Some of the supporters will spend all their remaining money for $c$, and others only part of it. Like in the algorithm, write $\money[i] = \frac{b}{n} - p_i(W)$ for the amount of money left over for $i$ at the current point in time. Rewriting the defining equation of $\rho$, we have that $c$ is $\rho$-affordable if and only if
\begin{align*}
	\sum_{i \in N}  u_i(c) \cdot  \min \left( \money[i]/u_i(c), \rho \right) = \cost(c) \text{.}
\end{align*}
Note that the supporters who spend all their remaining money will have the minimum attained in the first coordinate; hence these voters have the lowest values of $\money[i]/u_i(c)$. \Cref{alg:pseudocode} sorts the supporters of $c$ by this value, and then iterates through all $s \in [t]$ and checks whether the cost of $c$ can be covered with $i_1, \dots i_{s-1}$ spending all their remaining money and $i_s, \dots, i_t$ spending $\rho \cdot u_i(c)$ for some $\rho$. If that is possible, then we can choose $\rho$ such that
\begin{align*}
(\money[i_1] + \cdots + \money[i_{s-1}] ) + \sum_{j = s}^t  \rho \cdot u_{i_j}(c) = \cost(c) \iff 
\rho = \frac{\cost(c) - (\money[i_1] + \cdots + \money[i_{s-1}] )}{u_{i_{s}}(c) + \cdots + u_{i_t}(c)} \text{,}
\end{align*}
and this $\rho$ satisfies $\rho \cdot u_{i_s} \le \money[i_s]$. If this last condition fails, the algorithm needs to continue its iteration and try the next value for $s$. If the condition is satisfied, we have found the correct value of $\rho$.

This algorithm can easily be implemented to run in time $O(m^2n\log n)$, assuming we can do arithmetic in constant time. To achieve this time bound, we need to store the current values of the partial sums $\money[i_1] + \cdots + \money[i_{s}]$ and $u_{i_{s+1}}(c) + \cdots + u_{i_t}(c)$, so that in each iteration of the inner for-loop, we can update these sums in constant time. 
In future work, it would be of interest to design faster implementations of the method.

\begin{algorithm}
	\DontPrintSemicolon
	\SetAlgoNoEnd
	\SetAlgoLined
	\caption{Implementation of the Method of Equal Shares}
	\label{alg:pseudocode}
	$W \gets \emptyset$.\;
	For each voter $i \in N$, $\money[i]\gets b/|N|$\;
	\While{true}{
		\For{$c \in C \setminus W$}{
			\label{line:for-loop}
			\eIf{$\sum_{i \in N : u_i(c) > 0} \money[i] < \cost(c)$}{
				$\rho(c) \gets \infty$ (project $c$ is not affordable)\;
			}{
				Let $i_1, \dots, i_t$ be a list of all voters $i \in N$ with $u_i(c) > 0$, ordered so that $\money[i_1]/u_{i_1}(c) \le \cdots \le \money[i_t]/u_{i_t}(c) $.\;
				\label{line:sorting}
\For{$s = 1, \dots, t$}{
$\rho(c) \gets (\cost(c) - (\money[i_1] + \cdots + \money[i_{s-1}] ))/(u_{i_{s}}(c) + \cdots + u_{i_t}(c))$\;
					\If{$\rho(c) \cdot u_{i_s} \le \money[i_s]$}{
						\textbf{break} (we have found the $\rho$-value)\;
					}
				}
			}
		}
		\If{$\min_{c \in C \setminus W} \rho(c) = \infty$}{
			\Return{W}\;
		}
		$c \gets \argmin_{c \in C \setminus W} \rho(c)$ (break ties arbitrarily)\;
		$W \gets W \cup \{c\}$\;
		\For{$i \in N$ such that $u_i(c) > 0$}{
			$\money[i] \gets \money[i] - \min\{\money[i], \rho(c) \cdot u_i(c)\}$\;
		}
	}
\end{algorithm}

\subsection{Extended Justified Representation (EJR)}

The first notion of proportionality that we examine is Extended Justified Representation (EJR). This axiom was first proposed for approval-based committee elections~\citep{justifiedRepresentation}. Even for the special case of approval-based committee elections, only few rules are known to satisfy EJR~\citep{justifiedRepresentation,AEHLSS18,pet-sko:laminar}, but the Method of Equal Shares is one of them. In this section, we introduce a generalization of EJR to the PB model and show that our rule continues to satisfy EJR.

We first recall the definition of EJR for approval-based committee elections. Intuitively, this axiom ensures that every large enough group of voters whose approval sets have a large enough intersection must obtain a fair number of representatives. For example, if a group of voters forms an $\alpha$-fraction of the whole population and if this group agrees on sufficiently many candidates, then it should be allowed to decide about an $\alpha$-fraction of the elected candidates. Formally, this is achieved by excluding the possibility that each member of the group approves less than $\lfloor \alpha k \rfloor$ elected candidates.

\begin{definition}[Extended Justified Representation for approval-based committee elections]\label{def:ejr_approval_unit}
We say that a group of voters $S$ is \emph{$\ell$-cohesive} for $\ell\in\naturals$ if $|S| \geq \nicefrac{\ell}{k}\cdot n$ and $|\bigcap_{i \in S}A(i)| \geq \ell$.

A rule $\calR$ satisfies \emph{extended justified representation} if for each election instance $E$, each $\ell\in\naturals$, and each $\ell$-cohesive group $S$ of voters there exists a voter $i \in S$ such that $|A(i) \cap \calR(E)| \geq \ell$.
\end{definition} 
At first sight it is unintuitive that we only require that at least one voter obtains $\ell$ representatives. However, the strengthening of EJR that requires each member of $S$ to obtain $\ell$ representatives is impossible even on very small instances \citep{justifiedRepresentation}. Still, even with only the at-least-one guarantee, EJR has plenty of bite, in particular implying that the members of a cohesive group have high utility on average \citep{AEHLSS18,skowron:prop-degree,pet-sko:laminar}.

The generalization of this axiom to the PB model is not straightforward and to the best of our knowledge none has been proposed in the literature.\footnote{\citet{aziz2018proportionally} generalize the weaker axiom of Proportional Justified Representation (PJR) \citep{pjr17} beyond unit costs, but they operate in a non-standard utility model where voters care more about more expensive projects.}
To warm up, let's first relax the unit cost assumption, but stay in the approval-based setting. Then EJR should state the following.
\begin{definition}[Extended Justified Representation for approval-based elections]\label{def:ejr_approval}
	We say that a group of voters $S$ is $T$-cohesive for $T \subseteq C$ if $|S|/n \geq \cost(T)/b$ and $T \subseteq \bigcap_{i\in S} A(i)$.
A rule $\calR$ satisfies \emph{extended justified representation} if for each election instance $E$, each $T \subseteq C$, and each $T$-cohesive group $S$ of voters there exists a voter $i \in S$ such that $|A(i) \cap \calR(E)| \geq |T|$.
\end{definition}

Thus, cohesiveness now requires that the group $S$ can identify a collection of projects $T$ that they all approve and that is affordable with their fraction of the budget ($|S|/n \geq \cost(T)/b$). Note that voters $i \in S$ obtain utility $u_i(T) = |T|$ from $T$; EJR requires that at least one member of $S$ must attain this utility in the election outcome.

To further generalize EJR beyond approvals is more difficult, because the notion of a candidate who is approved by all members of $S$ does not have an analogue. Instead, we quantify cohesion by calculating the minimum utility that any member of $S$ assigns to each project in $T$.

\begin{definition}[Extended Justified Representation]\label{def:ejr_generalization}
	A group of voters $S$ is $(\alpha, T)$-\emph{cohesive}, where $\alpha\colon C \to [0;1]$ and $T \subseteq C$, if $|S|/n \geq \cost(T)/b$ and if $u_i(c) \geq \alpha(c)$ for  all $i\in S$ and $c \in T$.
A rule $\calR$ satisfies \myemph{extended justified representation} if for each election instance $E$, each $\alpha\colon C \to [0;1]$, $T \subseteq C$, and each $(\alpha, T)$-cohesive group of voters $S$ there exists a voter $i \in S$ such that $u_i(\calR(E)) \geq \sum_{c \in T}\alpha(c)$.
\end{definition}

Again, an $(\alpha, T)$-cohesive group of voters $S$ can propose the projects in $T$, since they are affordable with $S$'s share of the budget. The values $(\alpha(c))_{c\in T}$ denote how much the coalition $S$ agrees about the desirability of the projects in $T$. In particular, we have $u_i(T) \ge \sum_{c \in T}\alpha(c)$ for each $i \in S$. Consequently, \Cref{def:ejr_generalization} prohibits any outcome in which every voter in $S$ gets utility strictly lower than $\sum_{c \in T}\alpha(c)$; hence there must exist $i \in S$ such that $u_i(\calR(E)) \geq \sum_{c \in T}\alpha(c)$. 

EJR is a demanding property in the PB model. Consider the special case where there is only one voter, $N = \{1\}$. Then any outcome $W$ satisfying EJR must solve the knapsack problem, i.e., it must maximize $\sum_{c\in W} u_1(c)$ subject to the budget constraint, since otherwise an optimum knapsack $T$ witnesses an EJR violation. Because the knapsack problem is weakly NP-hard, this presents a difficulty for a rule to satisfy EJR.\footnote{Knapsack is also weakly NP-hard when assuming that the value of each item equals its weight (this is subset sum), so satisfying EJR is weakly NP-hard even for cost utilities, where every voter assigns a project $c$ utility either 0 or $\cost(c)$. \citet[Prop.~3.8]{aziz2018proportionally} prove a similar result for their BPJR-L notion.}

\begin{proposition}
	Unless P = NP, no aggregation rule that can be computed in strongly polynomial time can satisfy EJR in the general PB model.
\end{proposition}
Equal Shares can be computed in strongly polynomial time, and indeed it fails EJR in the general PB model.\footnote{\label{fn:ejr-counterexamples}Counterexamples: Suppose the budget is $b = 3$, there is just 1 voter, and two projects $c_1$ and $c_2$ with $\cost(c_1) = 1$ and $\cost(c_2) = 3$. Suppose the utilities are $u_1(c_1) = \smash{\frac12}$ and $u_1(c_2) = 1$. Then $c_1$ is $2$-affordable and $c_2$ is $3$-affordable. Thus Equal Shares selects $c_1$ and the outcome is $W = \{c_1\}$. (Note that this is exhaustive.) But the grand coalition $S = \{1\}$ can propose $T = \{c_2\}$ witnessing a violation of EJR. A counterexample with unit costs: $b = 2$, two voters, three projects, and utilities $u_1(c_1) = u_2(c_1) = 2$, $u_1(c_2) = u_2(c_3) = 3$, and $u_1(c_3) = u_2(c_2) = 0$. Then $c_1$ is $\smash{\frac14}$-affordable and $c_2$ and $c_3$ are both $\smash{\frac13}$-affordable. So $c_1$ is elected (with both voters paying equal amounts). Now nothing is affordable, and thus $W = \{c_1\}$. But then $S = \{1\}$ and $T = \{c_2\}$ is an EJR violation.}
Later we will prove that for every instance, an EJR outcome does exist (\Cref{thm:gcr-fjr}) but we do not know an efficiently computable method that finds such an outcome.
However, we can show that Equal Shares selects an outcome that satisfies a mild relaxation of EJR, which requires that EJR holds ``up to one project''. 

\begin{definition}[Extended Justified Representation Up To One Project]\label{def:ejr_up_to_1}
	\hspace{-6pt}\footnote{Compared to earlier versions of this paper, this definition of EJR up to one is slightly stronger: we now require that the ``extra'' project $c^*$ is a member of $T$ while previous versions allowed $c^*$ to be any project, not necessarily in $T$. As motivation for the new stronger version, consider a setting where there is some project $d$ that is too expensive to ever be affordable, but that provides very high utility to everyone. In those cases, the weaker version of EJR up to one does not provide any guarantees because we can choose $c^* = d$.\label{fn:stronger-EJR1}}
	\:A rule $\calR$ satisfies extended justified representation \myemph{up to one project} if for each election instance $E$ and each $(\alpha, T)$-cohesive group of voters $S$ there exists a voter $i^* \in S$ such that either $u_{i^*}(\calR(E)) \geq \sum_{c \in T}\alpha(c)$ or for some $c^* \in T$ it holds that $u_{i^*}(\calR(E)\cup \{c^*\}) > \sum_{c \in T}\alpha(c)$.
\end{definition}

It is worth noting that in the approval-based model, \Cref{def:ejr_generalization,def:ejr_up_to_1} are actually equivalent, because the ``up to one project'' option never applies: Consider an $(\alpha, T)$-cohesive group of voters $S$. Since voters' utilities are 0/1, we may assume that for each $c\in T$ we have $\alpha(c) = 1$: if $\alpha(c)>0$ this is clear; otherwise we can remove $c$ from $T$ without losing cohesiveness. Thus, the cohesiveness condition requires that every voter in $S$ approves every candidate in $T$. Finally, note that in the approval model, due to the strict inequality, both conditions $u_{i^*}(\calR(E)) \geq \sum_{c \in T}\alpha(c)$ and $\exists_{c^* \in T}. u_{i^*}(\calR(E)\cup \{c^*\}) > \sum_{c \in T}\alpha(c)$ boil down to $|A(i^*) \cap \calR(E)| \geq \sum_{c \in T}\alpha(c) = |T|$.

\begin{table}
	\centering
	\begin{tabular}{lll}
		\toprule
		& Approval utilities & Additive utilities \\
		\midrule 
		Unit costs & EJR & EJR up to one project \\
		\midrule
		General costs & EJR & EJR up to one project\textsuperscript{$\dagger$} \\
		\bottomrule
	\end{tabular}
	\caption{Equal Shares and Extended Justified Representation (see \Cref{thm:rule_x_and_ejr_cardinal} and \Cref{fn:ejr-counterexamples}). \\ $\dagger$: Unless P~=~NP, no strongly polynomial time method (such as Equal Shares) can satisfy EJR.}
	\label{tbl:equal-shares-ejr}
\end{table}

Our main result is that the Method of Equal Shares satisfies EJR up to one project in the general PB model. By the previous observation, it hence satisfies EJR in the approval-based model (even when not imposing unit costs), see \Cref{tbl:equal-shares-ejr}.

\begin{theorem}\label{thm:rule_x_and_ejr_cardinal}
	The Method of Equal Shares satisfies EJR up to one project in the participatory budgeting model.
\end{theorem}
\begin{proof}
	Let $S \subseteq N$ be a non-empty group of voters, and let $T \subseteq C$ be a proposal with $\cost(T)/b \le |S|/n$. For each $c \in C$, write $\alpha_c = \min_{i\in S} u_i(c)$, and write $\alpha = \sum_{c \in T} \alpha_c$. We assume that $\alpha_c > 0$ for all $c \in T$ (otherwise we can delete $c$ from $T$). If $W$ is the output of Equal Shares, we will show that there exists a voter $i^* \in S$ such that either $u_{i^*}(W) \ge \alpha$, or there is a candidate $c^* \in T$ such that $u_{i^*}(W \cup \{c^*\}) > \alpha$.
	
	In this proof, we will consider three runs of Equal Shares in different variations:
	\begin{enumerate}
		\item[(A)] Equal Shares run on the original instance (thus, outputting $W$).
		\item[(B)] Equal Shares run so that voters in $S$ are not bound by their budget constraint when paying for candidates in $T$. To make this formal, in the definition of Equal Shares, we redefine the notion of $\rho$-affordability so that $c \in T$ is $\rho$-affordable if
		\begin{equation}
			\label{eq:Baffordable-definition}
			\underbrace{\sum_{i \in S \vphantom{\setminus}} \rho \cdot u_i(c)}_{\text{no budget limit}} + \underbrace{\sum_{i \in N \setminus S} \min\{b_i, \rho \cdot u_i(c)\}}_{\text{with budget limit}} = \cost(c),
		\end{equation}
		and $c \in C \setminus T$ is $\rho$-affordable if
		\[
		\sum_{i \in S} \min\{ \underbrace{\max\{b_i, 0\}}_{\text{$b_i$ may be $<0$}}, \rho \cdot u_i(c)\} + \sum_{i \in N \setminus S} \min\{b_i, \rho \cdot u_i(c)\} = \cost(c),
		\]
		with payments defined as these equations suggest (namely, $i$'s payment is the value of the $i$th term of the sum).
		\item[(C)] Equal Shares run on a smaller instance where only candidates in $T$ and only voters in $S$ exist, and each voter has an unlimited budget $b_i = \infty$. In addition, we set $u_i(c) = \alpha_c$ for all $i \in S$ and $c \in T$.
	\end{enumerate}
	Note that in variations (B) and (C), all candidates in $T$ will be elected (eventually) as $\alpha_c > 0$ for all $c \in T$, and the voters in $S$ have unrestricted budgets when buying candidates in $T$.
	
	If at the end of the execution of (B) all voters in $S$ have spent strictly less than $b/n$, then (B) has selected all candidates in $T$ and no voter has overshot their budget. Thus (A) also elects all of $T$, so $u_i(W) \ge u_i(T) \ge \alpha$ for all $i \in S$, and we are done.
	Otherwise, let $i^* \in S$ be the first voter in $S$ who during the execution of (B) spends at least $b/n$.
	Suppose this happens just after (B) adds candidate $c^*$ to the outcome. Write $W_{\text{(B)}}$ for the set of candidates selected by (B) up to but excluding $c^*$. Note that (A) has also selected all candidates in $W_{\text{(B)}}$, because until that point the two rules behave identically. 
	
	\begin{figure}[ht]
		\centering
		\scalebox{0.77}{
			\begin{tikzpicture}
\draw[thick,->] (0.0, 0.0) -- (10.0, 0.0);
				\draw[thick,->] (0.0, 0.0) -- (0.0, 4.5);
				\node at (9.8, -0.3) {$x$};
				\node[anchor=west] at (0.1, 4.3) {$f_{\text{(B)}}(x)$};
				\node at (0, -0.3) {$0$};
				
				\foreach \x/\y [remember=\x as \lastx (initially 0), remember=\y as \lasty (initially 0)] in 
				{2.2/0.3, 3.4/0.6,5.4/1.6,7/3,8.8/3.4,9.5/4.2} {
\draw[help lines] (\lastx, \lasty) -- (\x, \lasty); \draw[help lines] (\x, \lasty) -- (\x, \y); \draw[-,thick,draw=red!70!black] (\lastx, \lasty) -- (\x, \y);
\node[fill=red!70!black, circle, inner sep=0.8pt] at (\x, \y) {};
				}
				
				\draw[decorate, decoration={brace,mirror,transform={yscale=1.1,yshift=-1pt}}] (3.45,0.6) -- (5.39,0.6);
				\node[font=\scriptsize] at (4.4, 0.25) {$u_i(d)$};
				\draw[decorate, decoration={brace,mirror,transform={yscale=1.1,yshift=-1.2pt}}] (5.4,0.63) -- (5.4,1.55);
				\node[font=\scriptsize] at (5.95, 1.1) {$p_i(d)$};
				
\draw[->, bend left] (7,3.9) to (7.8,3.4);
				\node[transform shape, scale=1, font=\footnotesize, align=right] at (4.9,3.9) {Can be non-convex for $d \not\in T$ \\ if $i$ has not enough budget left};
				
\draw[->, bend left] (3.4,1.7) to (4.2,1.2);
				\node[transform shape, scale=1, font=\footnotesize, align=right] at (2.45,1.7) {Slope $=\rho(d)$};
		\end{tikzpicture}}
		\scalebox{0.77}{
			\begin{tikzpicture}
\draw[thick,->] (0.0, 0.0) -- (10.0, 0.0);
				\draw[thick,->] (0.0, 0.0) -- (0.0, 4.5);
				\node at (9.8, -0.3) {$x$};
				\node[anchor=west] at (0.1, 4.3) {$f_{\text{(C)}}(x)$};
				\node at (0, -0.3) {$0$};
				
				\foreach \x/\y [remember=\x as \lastx (initially 0), remember=\y as \lasty (initially 0)] in 
{1.5/0.23, 3/0.65,5/1.9,6/2.8,6.8/3.8,7.1/4.5} {
\draw[help lines] (\lastx, \lasty) -- (\x, \lasty); \draw[help lines] (\x, \lasty) -- (\x, \y); \draw[-,thick,draw=blue!70!black] (\lastx, \lasty) -- (\x, \y);
\node[fill=blue!70!black, circle, inner sep=0.8pt] at (\x, \y) {};
				}
				
				\draw[decorate, decoration={brace,mirror,transform={yscale=1.1,yshift=-1pt}}] (3.05,0.65) -- (4.99,0.65);
				\node[font=\scriptsize] at (4.06, 0.25) {$\alpha_{c_3}$};
				\draw[decorate, decoration={brace,mirror,transform={yscale=1.1,yshift=-1.2pt}}] (5,0.67) -- (5,1.85);
				\node[font=\scriptsize] at (5.8, 1.23) {$\alpha_{c_3} \cdot \sigma_{c_3}$};
				
\draw[->, bend left] (3.2,1.8) to (3.9,1.35);
				\node[transform shape, scale=1, font=\footnotesize, align=right] at (2.33,1.8) {Slope $=\sigma_{c_3}$};
		\end{tikzpicture}}
		\caption{Illustration of functions $f_{\text{(B)}}(x)$ and $f_{\text{(C)}}(x)$.}
		\label{fig:fB-fC}
	\end{figure}
	
	Next, we will lower bound the utility that $i^*$ receives under (B) by the time $i^*$ has spent at least $b/n$. To do so, we define a function $f_{\text{(B)}}$ so that for a number $x$, $f_{\text{(B)}}(x)$ is the amount of money that $i^*$ had to spend during the execution of (B) until $i^*$ receives utility $x$. We make this into a continuous, piecewise linear function, so that to get a $\beta$-fraction of the utility of a candidate one needs to spend a $\beta$-fraction of the total spending for that candidate. See \Cref{fig:fB-fC}. In other words, $f_{\text{(B)}}$ consists of a sequence of line segments where the segment corresponding to $d \in C$ has length $u_i(d)$ and height $p_i(d)$ (the amount that $i^*$ paid for $d$). Note that the segment has a slope of usually $\rho(d)$, but it can be lower than $\rho(d)$ in case $d \not \in T$ and $i^*$'s budget was not enough to pay the full amount $\rho(d) \cdot u_{i^*}(d)$ for $d$.
	
	We can define a function $f_{\text{(C)}}$ in exactly the same way with respect to the execution of (C), based on the same voter $i^*$. The function $f_{\text{(C)}}$ is easy to understand. For each $c \in T$, let us write $\sigma_c = \cost(c)/(|S|\alpha_c)$, and let us label $T = \{c_1, \dots, c_r\}$ such that $\sigma_{c_1} \le \cdots \le \sigma_{c_r}$. Note that under (C) each not-yet-selected $c \in T$ is $\sigma_c$-affordable, because all voters have unlimited budgets and
	\[
	\sum_{i \in S} \sigma_c \cdot \alpha_c = |S| \cdot \sigma_c \cdot \alpha_c = \cost(c).
	\]
	It follows that $f_{\text{(C)}}$ consists of a sequence of line segments of length $\alpha_c$ and slope $\sigma_c$, one for each $c \in T$. These line segments come in increasing order of $\sigma_c$, i.e. in the order $c_1, \dots, c_r$, because Equal Shares always selects the $\rho$-affordable candidate with lowest $\rho$.
	
	We claim that 
	\begin{equation}
		\label{eq:fC-ge-fB}
		f_{\text{(C)}}(x) \ge f_{\text{(B)}}(x) \quad \text{for all $x \in [0, \alpha]$.}	
	\end{equation}
	(Intuitively, this says that under (C) the money of $i$ is used less efficiently for $i$ than under (B).) The inequality certainly holds at $x = 0$ because both functions take the value $0$. To establish \eqref{eq:fC-ge-fB} for other $x$, we will show that the slope of $f_{\text{(C)}}$ is always at least as high as the slope of $f_{\text{(B)}}$ (except of course when $x$ is a point joining two line segments, where the slope is not defined, but this only applies to finitely many points).
	
	Let us first note the following useful fact:
	\begin{equation}
		\label{eq:Caffordable-implies-Baffordable}
		\text{Under (B), at each step, any not-yet-selected $c \in T$ is $\rho$-affordable for some $\rho \le \sigma_c$.}
	\end{equation}
	Informally, the fact holds because there are extra voters in (B) compared to (C), and the voters in $S$ have weakly higher utility for $c$ in (B).
	Formally, looking at the definition \eqref{eq:Baffordable-definition} of affordability in (B), fact \eqref{eq:Caffordable-implies-Baffordable} follows because
	\[
	\sum_{i \in S} \sigma_c \cdot u_i(c) + \sum_{i \in N \setminus S} \min\{b_i, \sigma_c \cdot u_i(c)\} \ge \sum_{i \in S} \sigma_c \cdot u_i(c) \ge \sum_{i \in S} \sigma_c \cdot \alpha_c = \cost(c),
	\]
	where the last step holds because $c$ is $\sigma_c$-affordable during the execution of (C).
	
	\begin{figure}[ht]
		\centering
		\begin{tikzpicture}
\draw[thick,->] (0.0, 0.0) -- (11.3, 0.0);
			\draw[thick,->] (0.0, 0.0) -- (0.0, 4.5);
			\node at (11.1, -0.3) {$x$};
			\node[anchor=west, font=\footnotesize] at (8.25, 4.3) {$f_{\text{(B)}}(x)$};
			\node[anchor=west, font=\footnotesize] at (5.75, 4.3) {$f_{\text{(C)}}(x)$};
			\node at (0, -0.3) {$0$};
			
			\foreach \x/\y [remember=\x as \lastx (initially 0), remember=\y as \lasty (initially 0)] in 
			{2.2/0.3, 3.4/0.6,5.4/1.6,7/3,8.8/3.4,9.5/4.2} {
\draw[-,thick,draw=red!70!black!40!white] (\lastx, \lasty) -- (\x, \y);
\node[fill=red!70!black!40!white, circle, inner sep=0.8pt] at (\x, \y) {};
			}
			
			\foreach \x/\y [remember=\x as \lastx (initially 0), remember=\y as \lasty (initially 0)] in 
			{1.5/0.23, 3/0.65,5/1.9,6/2.8,6.8/3.8,7.1/4.5} {
\draw[-,thick,draw=blue!70!black!40!white] (\lastx, \lasty) -- (\x, \y);
\node[fill=blue!70!black!40!white, circle, inner sep=0.8pt] at (\x, \y) {};
			}
			
\draw[-,thick,draw=red!70!black] (5.4, 1.6) -- (7, 3);
			\node[fill=red!70!black, circle, inner sep=0.8pt] at (5.4, 1.6) {};
			\node[fill=red!70!black, circle, inner sep=0.8pt] at (7, 3) {};
			\draw[-,thick,draw=blue!70!black] (5, 1.9) -- (6, 2.8);
			\node[fill=blue!70!black, circle, inner sep=0.8pt] at (5, 1.9) {};
			\node[fill=blue!70!black, circle, inner sep=0.8pt] at (6, 2.8) {};
			
\draw[->, bend left] (7,1.9) to (6.3,2.2);
			\node[transform shape, scale=0.9, font=\footnotesize, align=right] at (7.7,1.9) {Segment 1};
			
\draw[->, bend left] (4.7,2.9) to (5.4,2.5);
			\node[transform shape, scale=0.9, font=\footnotesize, align=right] at (4,2.9) {Segment 2};
			
\draw (5.7,-0.1) -- (5.7,4.5);
			\node at (5.73, -0.25) {$x'$};
			
\draw[red!70!black!40!white, dashed] (5.4,-0.1) -- (5.4,1.6);
\draw[->, bend left] (4.9,0.3) to (5.3,0.1);
			\node[transform shape, scale=0.8, font=\footnotesize, align=right, anchor=east] at (4.9,0.3) {$i$'s utility at $t$};

\draw[blue!70!black!40!white, dashed] (6,-0.1) -- (6,2.8);
\draw[->, bend right] (6.5,0.3) to (6.1,0.1);
			\node[transform shape, scale=0.8, font=\footnotesize, align=right, anchor=west] at (6.5,0.3) {min.\ utility if $c_1,\dots,c_s$ all selected};
		\end{tikzpicture}
		\caption{Illustration of the proof of claim \eqref{eq:fC-ge-fB}.}
		\label{fig:segment-slopes}
	\end{figure}

	Now, let $x' \in [0, \alpha]$ be any point that is not a boundary point (for either $f_{\text{(B)}}$ or $f_{\text{(C)}}$). Say that $x'$ lies in the interior of the line segment corresponding to $d \in C$ of $f_{\text{(B)}}$ (call this Segment 1) and in the interior of the line segment corresponding to $c_s \in T$ of $f_{\text{(C)}}$ (call this Segment 2). See \Cref{fig:segment-slopes} for an illustration. Consider the time point $t$ when (B) chose to add $d$ to its outcome (but before it actually added $d$). At time $t$, $i^*$'s utility under (B) was equal to the $x$-coordinate of the left endpoint of Segment 1, and thus less than $x'$. Further, at time $t$, it cannot be the case that (B) has already selected all of the candidates $c_1, \dots, c_s$, because then $i^*$'s utility would be at least $\alpha_{c_1} + \dots + \alpha_{c_s}$ which is the $x$-coordinate of the right endpoint of Segment 2 and thus more than $x'$. Hence there is some $c_p$, $p \in [s]$, such that (B) has not selected $c_p$ before time $t$. 
	By fact \eqref{eq:Caffordable-implies-Baffordable}, $c_p$ is $\rho'$-affordable in (B) at time $t$ for some $\rho' \le \sigma_{c_p}$. Since (B) always selects a candidate that is $\rho$-affordable for the smallest $\rho$, it must be the case that $d$ is $\rho$-affordable for some $\rho \le \rho'$. Thus, the slope of Segment 1 is at most $\rho$ and hence at most $\sigma_{c_p}$. On the other hand, Segment 2 has slope $\sigma_{c_s}$. Note that $\sigma_{c_p} \le \sigma_{c_s}$ because $p \le s$. Thus Segment 1 has slope weakly lower than Segment 2. Since this is true for all $x'$ (not on boundary points), our claim \eqref{eq:fC-ge-fB} follows.

	Next, note that
	\[
	f_{\text{(C)}}(\alpha) = \sum_{c\in T} \sigma_c \cdot \alpha_c = \sum_{c\in T} \frac{\cost(c)}{|S|} = \frac{\cost(T)}{|S|} \le \frac{b}{n}.
	\]
	Using \eqref{eq:fC-ge-fB}, it follows that
	\begin{equation}
		\label{eq:fB-upper-bound}
		f_{\text{(B)}}(\alpha) \le \frac{b}{n}.
	\end{equation}
	
	Finally, consider the point in time just after (B) adds $c^*$ to its output. Voter $i^*$ has spent at least $b/n$ at this point. There are two cases:
	\begin{enumerate}
		\item[(i)] $i^*$ has spent exactly $b/n$. In this case $c^*$ is also selected by (A) because the rules behave identically until this point. Now \eqref{eq:fB-upper-bound} implies that $u_{i^*}(W_{\text{(B)}} \cup \{c^*\}) \ge \alpha$. Hence $u_{i^*}(W) \ge \alpha$.
		\item[(ii)] $i^*$ has spent strictly more than $b/n$. In this case, by definition of (B), we have $c^* \in T$. Now \eqref{eq:fB-upper-bound} implies that $u_{i^*}(W_{\text{(B)}} \cup \{c^*\}) > \alpha$. Because $W_{\text{(B)}} \subseteq W$, this implies $u_{i^*}(W \cup \{c^*\}) > \alpha$.
	\end{enumerate}
	In both cases, we conclude that $W$ satisfies EJR up to one project.
\end{proof}

\Cref{thm:rule_x_and_ejr_cardinal} establishes the Method of Equal Shares as a prime candidate for voting under a budget constraint, showing that it satisfies a demanding fairness property. This makes it the first known rule that can give such strong proportionality guarantees in the model with additive utilities and arbitrary costs. 

In the literature on the special case of approval-based committee elections, another rule has received much attention: Proportional Approval Voting (PAV). Let us briefly recall the definition of this rule.

\begin{definition}[Proportional Approval Voting (PAV)]
	For an approval-based election, \myemph{PAV} selects a feasible outcome maximizing $\sum_{i \in N} \mathrm{H}(|A(i) \cap W|)$,
	where $\mathrm{H}(r) = \sum_{j = 1}^r \nicefrac{1}{j}$ is the $r$th harmonic number.
\end{definition}

This rule satisfies EJR when assuming unit costs \citep{justifiedRepresentation}.
But without unit costs, the example of Onetown (\Cref{fig:equivalent-instances}) shows that PAV fails EJR (in the sense of \Cref{def:ejr_approval}). In fact, \Cref{ex:pav_no_ejr} shows that, for each $r \ge 0$, PAV does not even satisfy EJR up to $r$ projects.

\begin{example}
	[PAV fails EJR]
	\label{ex:pav_no_ejr}
	Let $r \in \naturals$ ($r \geq 2$), $b = r^3$, and consider the following approval-based profile:
	\begin{alignat*}{3}
		&\text{$r^2-1$ voters:} &\qquad& \{a_1, a_2, \ldots, a_r\}, \\
		&\text{$1$ voter:}      && \{b_1, b_2, \ldots, b_r\}.
	\end{alignat*}
	The candidates $a_1, a_2, \ldots, a_r$ cost $r^2$ dollars each; the candidates $b_1, b_2, \ldots, b_{r}$ cost $1$ dollar each. EJR requires that the one voter who approves candidates $b_1, \ldots, b_r$ must approve at least $r$ candidates in the outcome. However, PAV selects $\{a_1, a_2, \ldots, a_r\}$, leaving the voter with nothing. Note that this example is very simple: it is a so-called party list profile where all approval sets are either disjoint or equal.\qed
\end{example}

\subsection{Approximating the Core}

An important proportionality property that has been proposed for PB is the \emph{core}~\citep{justifiedRepresentation,fain2018fair}, an idea adapted from cooperative game theory.

\begin{definition}[The core]\label{def:core}
	For an election $(N, C, \cost, \{u_i\}_{i \in N})$, an outcome $W$ is in the \myemph{core} if for every $S \subseteq N$ and $T \subseteq C$ with $|S|/n \geq \cost(T)/b$ there exists $i \in S$ such that $u_i(W) \geq u_i(T)$.
\end{definition}

The core is a stronger guarantee than EJR. The core allows any group $S$ to present an arbitrary ``counter-proposal'' $T$ that they can afford, and guarantees that at least one member $i \in S$ would prefer to stick with the core outcome $W$, so $u_i(W) \geq u_i(T)$. EJR only guarantees that $u_i(W) \geq \sum_{c\in T} \min_{j\in S} u_j(c)$. Thus, EJR only respects counter-proposals $T$ if they have broad agreement within the coalition $S$. This is arguably a reasonable restriction, since such coalitions can more easily coordinate to ``complain'' against the selected $W$. Still, it would be nice to give the stronger core guarantee.
Unfortunately, there are elections where no outcome  is in the core, even with unit costs.

\begin{example}[The core may be empty]\!\!\!\footnote{This example is adapted from \citet[Appendix~C]{fain2018fair} so as to satisfy the unit cost assumption.} We have 6 voters and 6 candidates with unit costs, and $b = 3$. Utilities satisfy
	\begin{alignat*}{6}
		u_1(c_1) > u_1(c_2) > 0, \qquad
		&& u_2(c_2) > u_2(c_3) > 0, \qquad
		&& u_3(c_3) > u_3(c_1) > 0; \\
		u_4(c_4) > u_4(c_5) > 0, \qquad
		&& u_5(c_5) > u_5(c_6) > 0, \qquad
		&& u_6(c_6) > u_6(c_4) > 0,
	\end{alignat*}
	and all other utilities are equal to 0. Let $W \subseteq C$ be any feasible outcome, so $|W| \le 3$. Then either $|W \cap \{c_1, c_2, c_3\}| \le 1$ or $|W \cap \{c_4, c_5, c_6\}| \le 1$. Without loss of generality assume the former, and assume that $c_2 \not \in W$ and $c_3 \not\in W$. Then $S = \{v_2, v_3\}$ and $T = \{c_3\}$ block $W$, since $\frac13 = |S|/n \ge \cost(T) / b = \frac13$ and both $u_2(c_3) > u_2(c_1) \ge u_2(W)$ and $u_3(c_3) > u_3(c_1)\ge u_3(W)$. \qed
\end{example}

Notably, this example is not approval-based. It is unknown whether the core is always non-empty for approval-based elections (with or without the unit cost assumption).

In the committee context, \citet{pet-sko:laminar} showed that Equal Shares returns an outcome that never violates the core too badly.
We can generalize this result to the general PB setting: Equal Shares provides a multiplicative approximation to the core.\footnote{\citet{munagala2022approximate} study the same approximation notion and show that a constant factor 9.27-approximation always exists. A related approximation notion is studied by \citet{fain2018fair} (which also involves an additive term) and a different notion is studied by \citet{cheng2019group} and \citet{jiang2019approx} which approximates the coalition size.} 

\begin{definition}
	For $\alpha \ge 1$, we say that an outcome is in the \myemph{$\alpha$-core} if for every $S \subseteq N$ and $T \subseteq C$ with $|S|/n \geq \cost(T)/b$  there exists $i^* \in S$ and $c^* \in T$ such that $u_{i^*}(\calR(E) \cup \{c^*\}) \geq \frac{u_{i^*}(T)}{\alpha}$.
\end{definition}

\begin{theorem}\label{thm:rule_x_and_core_general_model}
	Given an election $E$, let $u_{\max}$ be the highest utility a voter can get from a feasible outcome. Let $u_{\min}$ we denote the smallest, yet positive utility a voter can get from a feasible outcome:
	\begin{align*}
		u_{\max} = \max_{i \in N} \max_{W:\cost(W) \leq 1} u_i(W) \qquad \text{and} \qquad u_{\min} = \min_{i \in N} \min_{W:u_i(W) > 0} u_i(W) \text{.}
	\end{align*}
	Then the outcome selected by Equal Shares is always in the $\alpha$-core for $\alpha = 4\log(2\cdot \nicefrac{u_{\max}}{u_{\min}})$.
\end{theorem}
\begin{proof}
For notational simplicity, without loss of generality we assume $b = 1$ (by scaling the costs of all projects appropriately).
Towards a contradiction, assume there exist an election instance $E$, a winning outcome $W \in \calR(E)$, a subset of voters $S \subseteq N$, and a subset of candidates $T \subseteq C$ with $\sum_{c \in T}\cost(c) \leq \nicefrac{|S|}{n}$ such that for every $i \in S$ and $c \in T$ it holds that $u_i(W \cup \{c\}) < \nicefrac{u_i(T)}{\alpha}$.

Now consider a fixed subset $S' \subseteq S$, and let
\begin{align*}
	\Delta(S') = \sum_{i \in S'} \big(u_i(T) - u_i(W)\big).
\end{align*}

Similarly as in the proof of \Cref{thm:rule_x_and_ejr_cardinal}, assume for the moment that the voters from $S'$ can spend more than their budgets when paying for candidates in $T$. Let us analyze how Equal Shares would proceed in this case. By the pigeonhole principle it follows that in each step of the rule, there exists a not-yet-elected candidate $c \in T \setminus W$ such that
\begin{align*}
	\frac{u_{S'}(c)}{\cost(c)} \geq \frac{\Delta(S')}{\cost(T)} \text{.}
\end{align*}
Indeed, if for each $c \in T \setminus W$ we had $\frac{u_{S'}(c)}{\cost(c)} < \frac{\Delta(S')}{\cost(T)}$, then
\begin{align*}
	\Delta(S') \leq \sum_{c \in T \setminus W} u_{S'}(c) < \sum_{c \in T \setminus W} \cost(c) \cdot \frac{\Delta(S')}{\cost(T)} \leq \Delta(S') \text{,}
\end{align*}
a contradiction.

Thus, in each step of the rule there exists a not-yet-elected candidate $c \in T \setminus W$ that is $\rho$-affordable for some $\rho \le \frac{\cost(T)}{\Delta(S')}$. Thus, in each step, Equal Shares selects some candidate $c \in C \setminus W$ that is $\rho$-affordable for some $\rho \le \frac{\cost(T)}{\Delta(S')}$, because Equal Shares always minimizes $\rho$.
Hence, the cost-per-utility that voters pay for the selected candidates is at most $\frac{\cost(T)}{\Delta(S')}$.
Now, consider the first moment when some voter in $S'$, say $i$, uses more than the voter's initial budget $\nicefrac{1}{n}$. Until this time moment, Equal Shares behaves exactly in the same way as if the voters from $S'$ had their initial budgets set to $\nicefrac{1}{n}$. Further, we know that in this moment, if we chose a candidate $c \in T$ that would be chosen if the voters had unrestricted budgets, then the voter $i$ would pay more than $\nicefrac{1}{n}$ in total, and thus, would get utility more than $\frac{1}{n} \cdot \frac{\Delta(S')}{\cost(T)}$. Since we assumed $u_i(W \cup \{c\}) < \nicefrac{u_i(T)}{\alpha}$, we get that
\begin{align*}
	\frac{u_i(T)}{\alpha} > u_i(W) + u_i(c) > \frac{1}{n} \cdot \frac{\Delta(S')}{\cost(T)} \text{.}
\end{align*}
Since $\alpha \geq 2$, and so $u_i(T) - u_i(W) \geq \nicefrac{u_i(T)}{2}$, we get that
\begin{align*}
	u_i(T) - u_i(W) \geq \frac{u_i(T)}{2} > \frac{\alpha \Delta(S')}{2n \cdot \cost(T)} \text{.}
\end{align*}
Let $S'' = S' \setminus \{i\}$. Clearly, we have
\begin{align*} 
	\Delta(S'') = \Delta(S') - (u_i(T)  - u_i(W)) \leq \Delta(S')\left(1 - \frac{\alpha}{2n \cdot \cost(T)}\right) \text{.}
\end{align*}
The above reasoning holds for each $S' \subseteq S$. Thus, we start with $S' = S$ and apply it recursively, in each iteration decreasing the size of $S'$ by $1$. After $\nicefrac{|S|}{2}$ iterations we are left with a subset~$S_{e}$ such that
\begin{align*} 
	\Delta(S_e) \leq \Delta(S)\left(1 - \frac{\alpha}{ 2n \cdot \cost(T)}\right)^{\frac{|S|}{2}} \leq \Delta(S)\left(1 - \frac{\alpha}{2 n \cdot \cost(T)}\right)^{\frac{\cost(T)n}{2}} < \Delta(S)\left(\frac{1}{e}\right)^{\frac{\alpha}{4}} \text{.}
\end{align*}  
Now, observe that $\Delta(S_e) \geq \nicefrac{|S|}{2} \cdot u_{\min}$ (for each $i \in S$ it must hold that $u_i(T) - u_i(W) > 0$) and that $ \Delta(S) \leq |S|\cdot u_{\max}$. Thus, we get that
\begin{align*}
	\frac{|S|}{2} u_{\min} \cdot e^{\frac{\alpha}{4}} < |S|\cdot u_{\max} \text{,}
\end{align*}
which is equivalent to $e^{\frac{\alpha}{4}} < 2\cdot \frac{u_{\max}}{u_{\min}}$ and, further, to $\alpha < 4\log(2\cdot \nicefrac{u_{\max}}{u_{\min}})$. This gives a contradiction and completes the proof.
\end{proof}

The bound of $\alpha = 4\log(2\cdot \nicefrac{u_{\max}}{u_{\min}})$ is asymptotically tight, and the hard instance can be constructed even for the approval-based committee-election model (there, $\nicefrac{u_{\max}}{u_{\min}} \leq b$). The precise construction is given by \citet{pet-sko:laminar}.

\subsection{Priceability of Equal Shares}\label{sec:priceability}

\citet{pet-sko:laminar} introduced a concept called priceability, which imposes a certain kind of balance on a voting rule. Every rule that, like the Method of Equal Shares, equally splits the budget between voters and then sequentially purchases projects using the money of its supporters will be priceable. Priceability does not place any restrictions on how the rule splits the project's cost among supporters. The concept also allows initial budgets higher than 1; an outcome is priceable if there exists \emph{some} budget limit for which it is priceable.

\begin{definition}[Priceability]\label{def:priceability}
	A price system is a pair $\textsf{ps} = (b', \{p_i\}_{i \in N})$, where $b' \geq b$ is the initial budget (where each voter controls equal part of the budget, namely $b'/n$), and for each voter $i \in N$, a \myemph{payment function} $p_i\colon C \to \reals_{\ge 0}$ that specifies the amount of money a particular voter pays for the elected candidates.\footnote{\citet{pet-sko:laminar} assumed that each voter is initially given one dollar, which corresponds to setting $b' = n$, but that there is an additional variable that specifies the total price that needs to be paid for an elected candidate. These two formulations are equivalent, but the present definition seems more natural for PB.}\textsuperscript{,}\footnote{The requirement that $b' \geq b$ ensures that the voters have at least enough money to buy candidates with total cost $b$, the real budget. Without this requirement, an empty outcome $W = \emptyset$ would be priceable (with $b' = 0$).} An outcome $W$ is supported by a price system $\textsf{ps} = (b', \{p_i\}_{i \in N})$ if the following conditions hold:
	\begin{description}
		\item[(C1).] $u_i(c) = 0 \implies p_i(c) = 0$ for each $i \in N$ and $c \in C$,\footnote{While condition (C1) is well-justified in the approval-based setting, for general additive valuations it is very weak. Indeed, the condition does not put any restrictions on the payments when $u_i(c)$ is very small but positive.}
		
		\item[(C2).] Each voter has the same initial budget of $\nicefrac{b'}{n}$: $\sum_{c \in C} p_i(c) \leq \nicefrac{b'}{n}$ for each $i \in N$.
		
		\item[(C3).] Each elected candidate is fully paid: $\sum_{i \in N}p_i(c) = \cost(c)$ for each $c \in W$.
		
		\item[(C4).] The voters do not pay for non-elected candidates: $\sum_{i \in N}p_i(c) = 0$ for each $c \notin W$.
		
		\item[(C5).] For each unelected candidate $c\not\in W$, the unspent budget of her supporters is at most~$\cost(c)$: $\sum_{i \in N : u_i(c) > 0} \left(\nicefrac{b'}{n} - \sum_{c' \in W}p_i(c')\right) \leq \cost(c)$ for each $c \notin W$.
	\end{description}
	
	An outcome $W$ is said to be \myemph{priceable} if there exists a price system $\textsf{ps} = (b', \{p_i\}_{i \in N})$ that supports $W$ (that satisfies conditions (C1)--(C5)).
\end{definition}

It is known that Equal Shares is priceable in the approval-based committee-election model \citep{pet-sko:laminar}, and in the general PB model this property is also preserved---indeed, the rule implicitly constructs the price system satisfying the above conditions.

\subsection{Exhaustiveness}\label{sec:exhaustiveness}

A basic and very desirable efficiency notion is \emph{exhaustiveness}, which requires that a voting rule spends its entire budget. Of course, due to the discrete model, we cannot guarantee that the rule will spend exactly 1 dollar (i.e., the entire budget); however, we can require that no additional project is affordable.

\begin{definition}[Exhaustiveness, \citealp{aziz2018proportionally}]\label{def:exhaustivness}
An election rule $\calR$ is \myemph{exhaustive} if for each election instance $E$ and each non-selected candidate $c \notin \calR(E)$ it holds that $\cost(\calR(E) \cup \{c\}) > b$. 
\end{definition}

Notably, Equal Shares fails to be exhaustive. It can happen that at the end of the execution of Equal Shares, some project remains affordable, but the project's supporters do not have enough money to pay for it, and so Equal Shares refuses to fund the project. For example, suppose that $b = 1$, that we have two voters and two candidates, such that $v_1$ approves $\{c_1\}$ and $v_2$ approves $\{c_2\}$. Suppose that both candidates have cost 1. Then Equal Shares returns $W = \emptyset$.

In fact, it turns out that exhaustiveness is incompatible with priceability.

\begin{example}[An instance where no exhaustive outcome is priceable]
	\label{ex:priceability_vs_exhaustiveness}
	Suppose that $b = 1$, we have 3 candidates and 3 voters. The first 2 voters approve $\{c_1\}$, and the third one approves $\{c_2, c_3\}$. We have $\cost(c_1) = 1$ and $\cost(c_2) = \cost(c_3) = \nicefrac{1}{3}$. The only exhaustive outcomes are $\{c_1\}$ and $\{c_2, c_3\}$. However, neither of them is priceable---indeed, to buy both $c_2$ and $c_3$, the third voter needs to control at least $\nicefrac{2}{3}$ dollars. Then the first two voters control  $\nicefrac{4}{3}$ dollars and can buy candidate $c_1$, a contradiction. On the other hand, to buy $c_1$, the first two voters need to control at least $1$ dollar. Then, the third voter controls at least $\nicefrac{1}{2}$ dollars and buys $c_2$ or $c_3$, a contradiction. \qed
\end{example}

In some contexts, it may actually be a desirable feature of Equal Shares that it is not exhaustive, especially if unspent budget can be used in other productive ways (such as in next year's PB election). Arguably, in non-exhaustive examples, no remaining project has sufficient support to justify its expense; on that view, no further projects should be funded. In other situations, unspent budget may not be reusable, such as when the budget comes from a grant where unspent money needs to be returned (and the relevant decision makers do not obtain value from the grant-maker's alternative activities), or when the `budget' is time (for example, when we use PB to plan activities for a day-long company retreat). In such situation, one might prefer an exhaustive rule.

\paragraph{Completion by varying the budget}
\citet{pet-pie-sha-sko:stable-priceability} suggested a method of completing outcomes returned by Equal Shares, which we call ``\emph{completion by varying the budget}''. In this variant, we evaluate the method with a budget $b' \ge b$ that is higher than the actually available budget $b$. We increase the value of $b'$ gradually (each time by some fixed amount), and after each increase we recompute the outcome from scratch using Equal Shares. We stop when we reach an exhaustive outcome or when the next increase of $b'$ would cause us to exceed the original budget $b$. The selected outcome is typically exhaustive, but formally there is no guarantee that the budget is exhausted. This is because the selected outcome will always be priceable and thus it may not be possible to be simultaneously exhaustive (\Cref{ex:priceability_vs_exhaustiveness}). For a  more detailed discussion, see also \citet{pet-pie-sha-sko:stable-priceability}.

An advantage of completion by varying the budget is that the voting power (in the form of the amount of budget initially distributed) continues to be equally shared, and in particular the outcome is priceable.
It also mirrors the way in which the d'Hondt apportionment divisor method ``completes'' lower quotas.

\paragraph{Completion by perturbation}
Since we have generalized Equal Shares to work for general additive valuations, there is another way for us to make it exhaustive. Recall that Equal Shares fails to be exhaustive in situations where the remaining projects' supporters do not have sufficient funds left. However, in elections where $u_i(c) > 0$ for all $i \in N$ and $C\in C$, every voter supports every candidate, and thus this problem never occurs. In fact, Equal Shares is exhaustive when run on profiles of this type.

\begin{proposition}\label{prop:exhaustive_rule_x}
	Consider an election $E = (N, C, \cost, \{u_i\}_{i \in N})$ such that $u_i(c) > 0$ for each $i \in N$ and  $c \in C$. The outcome returned by the Method of Equal Shares for $E$ is exhaustive.
\end{proposition}
\begin{proof}
	For the sake of contradiction assume that an outcome $W$ returned by Equal Shares for an election instance~$E$ is not exhaustive. Then, there exists a candidate $c \notin W$ such that $\cost(W \cup \{c\}) \leq b$. The voters paid in total $\cost(W)$ for $W$; their total initial budget was $b$, thus after $W$ is selected they all have at least $\cost(c)$ unspent money. However, this means that at the end of the execution of Equal Shares there exists a possibly very large value of $\rho$ such that
	\begin{align*}
		\sum_{i \in N} \min\left( \tfrac{b}{n} - p_i(W), u_i(c)\cdot\rho \right) =  \sum_{i \in N} \left( \tfrac{b}{n} - p_i(W)\right) \geq \cost(c) \text{.}
	\end{align*}
	Consequently, $c$ (or some other candidate) would be selected by Equal Shares, a contradiction.
\end{proof}

Thus, we can make Equal Shares exhaustive by perturbing the input utilities so that all utility values are positive. Specifically, for a small $\epsilon > 0$ ($\epsilon \ll \min_{u_i(c) > 0} u_i(c)$), and for each $i \in N, c \in C$ such that $u_i(c) = 0$ in the initial instance we set $u_i^\epsilon(c) = \epsilon$.
Then we run Equal Shares on the modified instance $\{u_i^\epsilon\}_{i \in N}$. By \Cref{prop:exhaustive_rule_x} the outcome is exhaustive. We call this strategy ``\emph{completion by perturbation}''.\footnote{By necessity, this rule sometimes elects candidates that cannot be afforded by its supporters (so it is not priceable). When we elect such a candidate $c$, this completion asks all supporters of $c$ to pay all their remaining money for $c$, and splits the remaining cost to be paid equally among voters who do not support $c$. Say that the maximum amount paid by a non-supporter for $c$ is $x$; this completion selects those candidates that minimize the value $x$ at each step. 

Another option is to set $u_i^\epsilon(c) = \cost(c) \cdot \epsilon$, which minimizes the fraction of a project's cost that is paid by non-supporters, rather than the absolute payment of non-supporters. Experiments suggest that this choice may be preferable.}

\paragraph{Other completions}
Users of Equal Shares could also choose to use the budget not spent by Equal Shares to further other goals. For example, they could use ``\emph{completion by utilitarian greedy}'' which selects unselected projects in order of their total utility $\sum_{i \in N} u_i(c)$, mirroring the voting rule commonly used in practice. \citet{pet-sko:laminar} proposed to complete Equal Shares with the use of  Phragm\'{e}n's sequential rule, but this only works for approval-based utilities.

\section{Greedy Cohesive Rule}

In \Cref{sec:rulex}, we discussed the EJR axiom for the PB model, and saw that it is implemented by Method of Equal Shares. 
We will now propose a strengthening of EJR, and show a rule that satisfies the new strong property. Interestingly, even in the approval-based committee-election model our new property is substantially stronger than EJR, and hence this new rule provides the strongest known proportionality guarantees. On the other hand, compared to Equal Shares, it is computationally expensive and arguably less natural.

\subsection{Full Justified Representation (FJR)}\label{sec:strong_ejr}

Our new proportionality axiom strengthens EJR by weakening its requirement that groups must be cohesive. Thus, the new axiom guarantees representation to groups that are only partially cohesive.

\begin{definition}[Full Justified Representation (FJR)]\label{def:ejr_generalization_strong}
We say that a group of voters $S$ is \emph{weakly $(\beta, T)$-cohesive} for $\beta\in \reals$ and $T \subseteq C$, if $|S|/n \geq \cost(T)/b$ and $u_i(T) \geq \beta$ for every voter $i\in S$.

A rule $\calR$ satisfies \emph{full justified representation (FJR)} if for each election instance $E$ and each weakly $(\beta, T)$-cohesive group of voters $S$ there exists a voter $i \in S$ such that $u_i(\calR(E))\geq \beta$.
\end{definition}

\begin{figure}[t]
	\centering
	\begin{tikzpicture}[xscale=2.8,every node/.style={inner sep=10pt,anchor=center}]
		\node (core) at (0,0) {Core};
		\node (fjr) at (1,0) {FJR};
		\node (ejr) at (2,0) {EJR};
		\node (pjr) at (3,0) {PJR};
		\node (jr) at (4,0) {JR};
		
		\begin{scope}[yshift=-0.7cm,every node/.style={font=\footnotesize, align=center}]
			\node at (1,0) {Full \\ Just.\ Repr.};
			\node at (2,0) {Extended \\ Just.\ Repr.};
			\node at (3,0) {Proportional \\ Just.\ Repr.};
			\node at (4,0) {Justified\\ Representation};
		\end{scope}
		
		\draw[->] 
			(core) edge (fjr)
			(fjr) edge (ejr)
			(ejr) edge (pjr)
			(pjr) edge (jr);
	\end{tikzpicture}
	\caption{Implications between different Justified Representation axioms}
	\label{fig:jr-hasse}
\end{figure}

In the approval-based committee-election model, FJR boils down to the following requirement: Let $S$ be a group of voters, and suppose that each member of $S$ approves at least $\beta$ candidates from some set $T\subseteq C$ with $|T| \le \ell$, and let $|S| \geq \nicefrac{\ell}{b}\cdot n$. Then at least one voter from $S$ must have at least $\beta$ representatives in the committee. It is clear that in the special case of $\beta=\ell$, we obtain \Cref{def:ejr_approval_unit}, hence FJR implies EJR. The same implication holds in the general PB model.

\begin{proposition}\label{prop:fjr_implies_ejr}
	FJR implies EJR in the general PB model.
\end{proposition}
\begin{proof}
Suppose that rule $\calR$ satisfies FJR and take an $(\alpha, T)$-cohesive group of voters $S$ for some $\alpha\colon T \to [0;1]$, $T \subseteq C$. For every voter $i\in S$ and every candidate $c\in T$ we have $u_i(c) \geq \alpha(c)$. We set $\beta = \sum_{c\in T}\alpha(c)$; clearly, we have also $u_i(T) \geq \beta$, thus $S$ is weakly cohesive. As $\calR$ satisfies FJR, we have that $u_i(\calR(E))\geq \beta = \sum_{c\in T}\alpha(c)$, which completes the proof.
\end{proof}

In turn, it is easy to see that FJR is implied by the core property (cf. \Cref{def:core}). \Cref{fig:jr-hasse} shows the logical relationships between different JR-type axioms.
An adaptation of the PJR axiom to the general PB model was proposed by \citet{los2022systematization}.
FJR is related to, but stronger than, some other relaxations of the core discussed by \citet[Section~5.2]{pet-sko:laminar}.
The two major rules known to satisfy EJR for approval-based committee elections (Equal Shares and PAV) both fail FJR; to the best of our knowledge, no known rule satisfies FJR for approval-based committee elections, let alone for the general PB model.

\begin{example}[Equal Shares fails FJR\footnote{An open question, asked by Jannik Peters, is whether the Equal Shares outcome can always be completed to an outcome satisfying FJR. We believe the answer is no.}]\label{ex:rulex_no_strong_ejr}
Consider the following instance of approval-based committee elections for $n=22$ voters, $m = 13$ candidates, and where the goal is to select a committee of size $b=11$:
\begin{align*}
&\text{voters 1-3}\colon   &&\{c_1, c_2, c_3, c_4, c_8\}    && \text{voters 13-15}\colon && \{c_1, c_2, c_3, c_4, c_{12}\}\\
&\text{voters 4-6}\colon   &&\{c_1, c_2, c_3, c_4, c_9\}    && \text{voters 16-18}\colon && \{c_5, c_6, c_7, c_8, c_9, c_{10}, c_{11}, c_{12}\}\\
&\text{voters 7-9}\colon   &&\{c_1, c_2, c_3, c_4, c_{10}\} && \text{voters 19-21}\colon && \{c_5, c_6, c_7\}\\
&\text{voters 10-12}\colon &&\{c_1, c_2, c_3, c_4, c_{11}\} && \text{voter 22}\colon     && \{c_{13}\} \text{.}
\end{align*}
In the first 4 steps, Equal Shares chooses candidates $c_1, c_2, c_3, c_4$ (this happens for $\rho=\nicefrac{1}{15}$). After that, each of the first 15 voters has a remaining budget of  $\nicefrac{11}{22}-\nicefrac{4}{15}$. In next 3 steps, for $\rho=\nicefrac{1}{6}$, candidates $c_5, c_6, c_7$ are chosen: the 6 voters who support them spend all their money ($\nicefrac{11}{22} - 3 \cdot \nicefrac{1}{6} = 0$). After that, the algorithm stops. Each of the first $15$ voters has $4$ candidates she approves; voters 16-18 approve 3 selected candidates. Thus, no member of the weakly $(5,\{c_1, c_2, c_3, c_4,c_8, c_9, c_{10}, c_{11},c_{12}\})$-cohesive group of the first 18 voters has 5 representatives. \qed
\end{example}

\begin{example}[PAV fails FJR]\label{ex:pav_no_strong_ejr}
	This example was first considered by \citet[Section~1]{pet-sko:laminar}. We have $m = 15$ candidates and $n = 6$ voters, with the following preferences: 
	
	\vspace{8pt}
	\begin{minipage}{0.65\textwidth}
	\vspace{-8pt}\begin{align*}
		&\text{voter 1}\colon && \{c_1, c_2, c_3, c_4\} &\quad& \text{voter 4}\colon && \{c_7, c_8, c_9\} \\
		&\text{voter 2}\colon && \{c_1, c_2, c_3, c_5\} && \text{voter 5}\colon && \{c_{10}, c_{11}, c_{12}\} \\
		&\text{voter 3}\colon && \{c_1, c_2, c_3, c_6\} && \text{voter 6}\colon && \{c_{13}, c_{14}, c_{15}\} \text{.} \\
	\end{align*}
	\end{minipage}\quad
	\begin{minipage}{0.3\textwidth}
	\vspace{0pt} \scalebox{0.95}{
	\begin{tikzpicture}
		[yscale=0.43,xscale=0.78,voter/.style={anchor=south, yshift=-7pt}, select/.style={fill=blue!10}, c/.style={anchor=south, yshift=1.5pt, inner sep=0}]
		\draw[select] (0,0) rectangle (3,1);
		\draw[select] (0,1) rectangle (3,2);
		\draw[select] (0,2) rectangle (3,3);
		\draw (0,3) rectangle (1,4);
		\draw (1,3) rectangle (2,4);
		\draw (2,3) rectangle (3,4);
		\node at (1.5,0.42) {$c_1$};
		\node at (1.5,1.42) {$c_2$};
		\node at (1.5,2.42) {$c_3$};
		\node at (0.5,3.42) {$c_4$};
		\node at (1.5,3.42) {$c_5$};
		\node at (2.5,3.42) {$c_6$};
		\foreach \x in {3,4,5}
		{
			\foreach \y in {0,1,2}
			{
				\draw[select] (\x,\y) rectangle (\x+1,\y+1);
			}
		}
		\node at (3.5,0.42) {$c_{7}$};
		\node at (3.5,1.42) {$c_{8}$};
		\node at (3.5,2.42) {$c_{9}$};
		\node at (4.5,0.42) {$c_{10}$};
		\node at (4.5,1.42) {$c_{11}$};
		\node at (4.5,2.42) {$c_{12}$};
		\node at (5.5,0.42) {$c_{13}$};
		\node at (5.5,1.42) {$c_{14}$};
		\node at (5.5,2.42) {$c_{15}$};
		\foreach \i in {1,...,6}
		\node[voter] at (\i-0.5,-0.6) {$v_\i$};
	\end{tikzpicture}}
	\end{minipage}

	\noindent
	The size of the committee is $k=12$. PAV chooses the committee $\{c_1, c_2, c_3, c_7, c_8, c_9, c_{10}, c_{11}, c_{12}, c_{13}, c_{14}, c_{15}$\}. Hence, no voter from the weakly $(4, \{c_1, c_2, c_3, c_4, c_5, c_6\})$-cohesive group consisting of the first 3 voters has 4 representatives. \qed
\end{example}

Still, it turns out that FJR can always be satisfied: we present a (somewhat artificial) rule satisfying this strong notion of proportionality.\footnote{A similar rule is used by \citet[Proposition 3.7]{aziz2018proportionally} to prove that there always exists an outcome satisfying their axiom BPJR-L for approval-based instances. That axiom is weaker than FJR applied to a utility profile where $u_i(c) = \cost(c)$ whenever $i$ approves $c$, and $u_i(c) = 0$ otherwise.}

\begin{definition}[Greedy Cohesive Rule (GCR)]\label{def:gcr}
The \emph{Greedy Cohesive Rule} (GCR) is defined sequentially as follows: we start with an empty outcome $W=\emptyset$ and label all voters as \emph{active}. At each step, we search for a value $\beta > 0$, a set of voters $S\subseteq N$ who are all active, and a set of candidates $T\subseteq C \setminus W$ such that $S$ is weakly $(\beta, T)$-cohesive. If such values of $\beta$, $S$, and $T$ do not exist, then we stop and return $W$. Otherwise, we pick values of $\beta$, $S$, and $T$ that maximize $\beta$, breaking ties in favor of smaller $\cost(T)$.\footnote{This way of breaking ties matters in \Cref{sec:gcr-extendable} (to show that GCR can be extended to a priceable outcome), but it does not matter how we break ties for the proof of \Cref{thm:gcr-fjr} (to show that GCR satisfies FJR).}
We add all the candidates from $T$ to $W$, and then label all voters in $S$ as inactive.
\end{definition}

Let us first check that the Greedy Cohesive Rule always selects an outcome that does not exceed the budget limit (which we have normalized to 1). Indeed, whenever the algorithm adds some set $T$ to $W$, then by definition of weakly cohesive groups, we have $|S|/n \geq \cost(T)/b$, and hence the algorithm labels at least $\cost(T)/b\cdot n$ voters as inactive after this step. Thus, if GCR selects an outcome with total cost  $\cost(W)$, then it must have inactivated at least $\cost(W)/b\cdot n$ voters during its execution. Because there are $n$ voters, we have $\cost(W)/b \cdot n \le n$, and hence $\cost(W) \le b$.

\begin{theorem}
	\label{thm:gcr-fjr}
	The Greedy Cohesive Rule satisfies FJR.
\end{theorem}
\begin{proof}
	Assume for a contradiction that there exists a weakly $(\beta, T)$-cohesive group $S$ which witnesses that FJR is not satisfied by the outcome selected by GCR. Consider the voter $i \in S$ who was first labeled inactive by GCR, and the outcome $W$ right after that step (since $S$ is weakly cohesive, such an $i$ always exists). Since $i \in S$ and $S$ witnesses the FJR failure, we have $u_i(W) < \beta$. We know that $i$ was inactivated as a member of some weakly $(\beta', T')$-cohesive group $S'$. Just before $S'$ was labeled inactive, all of the members of $S$ were active. Thus, we have $\beta' \geq \beta$, as GCR maximizes this value. However, since $T' \subseteq W$, we have 
	\[ \beta' \leq u_i(T') \leq u_i(W) < \beta, \] 
	a contradiction to $\beta' \geq \beta$. Hence, such a group $S$ does not exist.
\end{proof}

\begin{remark}[GCR and FJR do not require additivity]
	In this paper, we have assumed additive utility functions throughout. But the definitions of FJR and of GCR make sense for any utility functions $u_i : 2^C \to \mathbb{R}_{\ge 0}$ over subsets of $C$ that are \emph{monotone} in the sense that $u_i(T_1) \ge u_i(T_2)$ whenever $T_1 \supseteq T_2$. Monotone utility functions allow us to encode, for example, complementarities and substitutes. The proof of \Cref{thm:gcr-fjr} only used monotonicity, and hence GCR satisfies FJR for all monotone utility functions. In contrast, the definition of EJR relies on additivity.
	\qed
\end{remark}

\subsection{Priceability and Exhaustiveness of GCR}
\label{sec:gcr-extendable}
GCR satisfies neither priceability (\Cref{def:priceability}) nor exhaustiveness (\Cref{def:exhaustivness}). However, we will prove that an outcome elected by this rule can always be completed to a priceable one; this suggests that GCR never elects outcomes that are ``too unbalanced''. In the proof of \Cref{thm:gcr_priceable} we describe precisely how such a completion can be implemented. Using a somewhat different completion scheme, we can complete GCR to an exhaustive outcome. This way we obtain an outcome that is both exhaustive and also close to being priceable.

We start by proving two useful lemmas, which establish a kind of ``Hall condition'' for priceability which may of independent interest.

\begin{lemma}\label{fact-1}
	Let $S$ be a $(\beta, T)$-cohesive group which is selected in some step of GCR. For every subset $A \subseteq T$, the size of the set of voters $S':=\{i \in S\colon u_i(A) > 0\}$ is at least $\cost(A)\cdot \frac{n}{b}$.
\end{lemma}
\begin{proof}
	The statement is trivial for $\cost(A) = 0$, so assume that $\cost(A) > 0$. Assume for a contradiction that the set $S'\subseteq S$ defined above has size $|S'| < \cost(A)\cdot \frac{n}{b}$. Then the group $S \setminus S'$ together with the set $T \setminus A$ is $(\beta, T \setminus A)$-cohesive because
	\[
		\textstyle
		|S \setminus S'| > |S| - \cost(A)\cdot \frac{n}{b} \geq \cost(T)\cdot \frac{n}{b} - \cost(A)\cdot \frac{n}{b} = \cost(T \setminus A)\cdot \frac{n}{b} \text{.}
	\]
	Further, as $\cost(A) > 0$, we have $\cost(T \setminus A) < \cost(T)$. Thus, GCR would select $S\setminus S'$ instead of $S$, a contradiction.
\end{proof}

\begin{lemma}\label{lemma-gcr-priceable}
	For every outcome of the GCR rule, there always exists a payment function satisfying conditions (C1)--(C4) with $b'=b$.
\end{lemma}
\begin{proof}
	Consider a single step of GCR and let $S$ be a $(\beta, T)$-cohesive group considered in that step. We will prove that there exists a price system in which voters from $S$ pay $\cost(c)$ dollars for each candidate $c\in T$. By combining these price systems for all the (pairwise disjoint) groups $S$ selected by GCR, we obtain a price system for the  outcome of GCR.

	Without loss of generality, we assume that $\frac{b}{n}$ is an integer and that $\cost(c)$ is an integer for all $c \in C$ (we can just appropriately scale up all costs and the budget, because we assumed that the budget and the costs are rational numbers). 
	
	We now imagine that each candidate $c \in T$ is split into $\cost(c)$ many \emph{pieces}, each with cost $1$. Let $A_T$ be the set of all pieces. We also imagine that each voter $i \in S$ has $\frac{b}{n}$ many \emph{coins}, each worth $1$, and let $A_S$ be the set of all coins. Note that one coin can pay for one piece.

	Consider the bipartite graph $G = (A_S + A_T, E)$, where 
there is an edge between a coin belonging to voter $i\in S$ and a piece of a candidate $c\in T$ if and only if $u_i(c) > 0$. 
	
	Now, consider any subset of pieces $A\subseteq A_T$, and let us assess the size of the neighbourhood $N_G(A) \subseteq A_S$. Let $C(A)$ denote the set of candidates who have at least one piece in $A$. Then $N_G(A)$ consists of all the coins of those voters who assign a positive utility to some candidate from $C(A)$. By \Cref{fact-1} there are at least $\cost(C(A)) \cdot \frac{n}{b}$ such voters, each of whom have $\frac{b}{n}$ many coins. Thus
\[
	\textstyle
	|N_G(A)| \ge \cost(C(A)) \cdot \frac{n}{b} \cdot \frac{b}{n} = \cost(C(A))  \ge |A|.
	\]
Hence, by Hall's theorem, there is a one-to-one mapping between coins and pieces. This allows us to construct payment functions as follows:
	For every voter $i\in S$ and candidate $c \in T$, if exactly $q$ coins of $i$ are mapped to some parts of $c$, then $p_i(c) = q$. It is straightforward to check that such a payment function satisfies conditions (C1)--(C4) for $b'=b$, which completes the proof.
\end{proof}

Finally, we can state the main result of this subsection.

\begin{theorem}\label{thm:gcr_priceable}
	Every outcome $W$ elected by GCR can be completed to some priceable outcome.
\end{theorem}
\begin{proof}
	From \Cref{lemma-gcr-priceable}, we know that there exists a family of payment functions $\{p\}_{i \in N}$ satisfying conditions (C1)--(C4) for $W$. Now, to obtain outcome $W'\supseteq W$ supported by a valid price system, it is enough to run Equal Shares for this instance with the initial outcome set to $W$ and the initial endowment of each voter $i\in N$ set to $\nicefrac{b}{n} - \sum_{c \in W} p_i(c)$.
\end{proof}

\subsection{Some Drawbacks of the Greedy Cohesive Rule}
\label{sec:gcr-drawbacks}
Since GCR satisfies FJR but Equal Shares does not, we may conclude that GCR is a better rule. Clearly, GCR is custom-engineered to satisfy FJR. Thus, we may expect the rule to be deficient in other dimensions. The results presented in \Cref{sec:gcr-extendable} certainly suggest that GCR is not pathological, but in this section we consider some examples where Equal Shares seems to select better outcomes than GCR.

We begin by discussing a property that \citet{pet-sko:laminar} call \myemph{laminar proportionality}. This property identifies a family of well-behaved preference profiles and specifies the outcome on those profiles. The axiom is defined for the case of approval-based committee elections. Equal Shares satisfies it; the following example shows that GCR does not.

\begin{example}[GCR fails laminar proportionality]
	\label{ex:gcr_not_laminar}
	Let $N=\{1,2,3,4\}$ and $b=8$, and introduce the candidate sets $X = \{c_1, \ldots, c_4\}$, $Y=\{c_5, \ldots, c_{10}\}$, and $Z=\{c_{11}, c_{12}\}$. The first three voters approve $X \cup Y$, and the fourth one approves $X \cup Z$. Two copies of the profile are depicted below. The candidates are represented by boxes; each candidate is approved by the voters who are below the corresponding box.

	\begin{center}
		
		\begin{minipage}{0.35\linewidth}
			\begin{center}
				\begin{tikzpicture}
					\node at (0.5, -0.25) {$v_1$};
					\node at (1.5, -0.25) {$v_2$};
					\node at (2.5, -0.25) {$v_3$};
					\node at (3.5, -0.25) {$v_4$};
					
					\filldraw[fill=green!10!white, draw=black] (0,0.0) rectangle ++(4.0,0.5);
					\node at (2.0, 0.25) {$c_{1}$};
					\filldraw[fill=green!10!white, draw=black] (0,0.5) rectangle ++(4.0,0.5);
					\node at (2.0, 0.75) {$c_{2}$};
					\filldraw[fill=green!10!white, draw=black] (0,1.0) rectangle ++(4.0,0.5);
					\node at (2.0, 1.25) {$c_{3}$};
					\filldraw[fill=green!10!white, draw=black] (0,1.5) rectangle ++(4.0,0.5);
					\node at (2.0, 1.75) {$c_{4}$};
					
					\filldraw[fill=green!10!white, draw=black] (0,2.0) rectangle ++(3.0,0.5);
					\node at (1.5, 2.25) {$c_{5}$};
					\filldraw[fill=green!10!white, draw=black] (0,2.5) rectangle ++(3.0,0.5);
					\node at (1.5, 2.75) {$c_{6}$};
					\filldraw[fill=green!10!white, draw=black] (0,3.0) rectangle ++(3.0,0.5);
					\node at (1.5, 3.25) {$c_{7}$};
					\filldraw[fill=white, draw=black] (0,3.5) rectangle ++(3.0,0.5);
					\node at (1.5, 3.75) {$c_{8}$};
					\filldraw[fill=white, draw=black] (0,4.0) rectangle ++(3.0,0.5);
					\node at (1.5, 4.25) {$c_{9}$};
					\filldraw[fill=white, draw=black] (0,4.5) rectangle ++(3.0,0.5);
					\node at (1.5, 4.75) {$c_{10}$};
					
					\filldraw[fill=green!10!white, draw=black] (3.0,2.0) rectangle ++(1.0,0.5);
					\node at (3.5, 2.25) {$c_{11}$};
					\filldraw[fill=white, draw=black] (3.0,2.5) rectangle ++(1.0,0.5);
					\node at (3.5, 2.75) {$c_{12}$};
					
				\end{tikzpicture}
			\end{center}
		\end{minipage}
		\qquad
		\begin{minipage}{0.35\linewidth}
			\begin{center}
				\begin{tikzpicture}
					\node at (0.5, -0.25) {$v_1$};
					\node at (1.5, -0.25) {$v_2$};
					\node at (2.5, -0.25) {$v_3$};
					\node at (3.5, -0.25) {$v_4$};
					
					\filldraw[fill=white, draw=black] (0,0.0) rectangle ++(4.0,0.5);
					\node at (2.0, 0.25) {$c_{1}$};
					\filldraw[fill=white, draw=black] (0,0.5) rectangle ++(4.0,0.5);
					\node at (2.0, 0.75) {$c_{2}$};
					\filldraw[fill=white, draw=black] (0,1.0) rectangle ++(4.0,0.5);
					\node at (2.0, 1.25) {$c_{3}$};
					\filldraw[fill=white, draw=black] (0,1.5) rectangle ++(4.0,0.5);
					\node at (2.0, 1.75) {$c_{4}$};
					
					\filldraw[fill=blue!10!white, draw=black] (0,2.0) rectangle ++(3.0,0.5);
					\node at (1.5, 2.25) {$c_{5}$};
					\filldraw[fill=blue!10!white, draw=black] (0,2.5) rectangle ++(3.0,0.5);
					\node at (1.5, 2.75) {$c_{6}$};
					\filldraw[fill=blue!10!white, draw=black] (0,3.0) rectangle ++(3.0,0.5);
					\node at (1.5, 3.25) {$c_{7}$};
					\filldraw[fill=blue!10!white, draw=black] (0,3.5) rectangle ++(3.0,0.5);
					\node at (1.5, 3.75) {$c_{8}$};
					\filldraw[fill=blue!10!white, draw=black] (0,4.0) rectangle ++(3.0,0.5);
					\node at (1.5, 4.25) {$c_{9}$};
					\filldraw[fill=blue!10!white, draw=black] (0,4.5) rectangle ++(3.0,0.5);
					\node at (1.5, 4.75) {$c_{10}$};
					
					\filldraw[fill=blue!10!white, draw=black] (3.0,2.0) rectangle ++(1.0,0.5);
					\node at (3.5, 2.25) {$c_{11}$};
					\filldraw[fill=blue!10!white, draw=black] (3.0,2.5) rectangle ++(1.0,0.5);
					\node at (3.5, 2.75) {$c_{12}$};
				\end{tikzpicture}
			\end{center}
		\end{minipage}
		
	\end{center}
	
	In this election instance, laminar proportionality would require that the voting rule selects all the candidates from $X$ since they are approved by everyone. After electing the candidates in $X$, four seats are left to fill. Since the group $\{v_1,v_2,v_3\}$ the three times as large as the group $\{v_4\}$, laminar proportionality requires that we elect three candidates from $Y$ and one candidate from $Z$. Thus, the committee indicated by the green boxes on the left-hand figure is laminar proportional.
	
	On the other hand, in the first step GCR can choose the weakly $(6, Y)$-cohesive group $\{v_1,v_2,v_3\}$ and in the second step it can select the weakly $(2,Z)$-cohesive group $\{v_4\}$. This results in the blue committee depicted in the right-hand figure; this committee fails laminar proportionality. \qed
\end{example}

\Cref{ex:gcr_not_laminar} shows that in general, GCR is not laminar proportional, as it can return committees which are prohibited by the axiom. However, this example is not fully satisfactory, as it depends on tie-breaking. For example, in the first step we could choose the weakly $(6, \{c_2, c_3, c_4, c_5, c_6\})$-cohesive group containing the first three voters, and in the second step the weakly $(2,\{c_1, c_{11}\})$-cohesive group containing the last voter. An open question is whether GCR can always elect a committee satisfying laminar proportionality (among others). However, the following example shows that for some 'nearly laminar' instances, GCR does not match the general intuition standing behind this axiom.

\begin{example}\label{ex:gcr_not_nearly_laminar}
	Modify the instance described in \Cref{ex:gcr_not_laminar} in the following way: we have $N=[4000]$. Voter $1$ approves only candidates from $Y$, voters $2$ to $3000$ approve $X \cup Y$, voters $3001$ to $3999$ approve $X \cup Z$ and voter $4000$ approves $Z$. 
	\begin{center}
		\begin{tikzpicture}[xscale=1.7]
			\node[font=\footnotesize] at (0.125, -0.5) {$\{v_1\}$};
			\node[font=\footnotesize] at (3.125, -0.5) {$\{v_2, \dots, v_{3000}\}$};
			\node[font=\footnotesize] at (6.875, -0.5) {$\{v_{3001},\dots,v_{3999}\}$};
			\node[font=\footnotesize] at (7.875, -0.5) {$\{v_{4000}\}$};
			
			\draw[decorate,decoration={brace,mirror,transform={yscale=1.5}}] (0, -0.1) -- (0.24,-0.1);	
			\draw[decorate,decoration={brace,mirror,transform={yscale=1.5}}] (0.26, -0.1) -- (5.99,-0.1);			
			\draw[decorate,decoration={brace,mirror,transform={yscale=1.5}}] (6.01, -0.1) -- (7.74,-0.1);	
			\draw[decorate,decoration={brace,mirror,transform={yscale=1.5}}] (7.76, -0.1) -- (8,-0.1);	
			
			\filldraw[fill=white, draw=black] (0.25,0.0) rectangle ++(7.5,0.5);
			\node at (4.0, 0.25) {$c_{1}$};
			\filldraw[fill=white, draw=black] (0.25,0.5) rectangle ++(7.5,0.5);
			\node at (4.0, 0.75) {$c_{2}$};
			\filldraw[fill=white, draw=black] (0.25,1.0) rectangle ++(7.5,0.5);
			\node at (4.0, 1.25) {$c_{3}$};
			\filldraw[fill=white, draw=black] (0.25,1.5) rectangle ++(7.5,0.5);
			\node at (4.0, 1.75) {$c_{4}$};
			
			\filldraw[fill=blue!10!white, draw=black] (0,2.0) rectangle ++(6.0,0.5);
			\node at (3, 2.25) {$c_{5}$};
			\filldraw[fill=blue!10!white, draw=black] (0,2.5) rectangle ++(6.0,0.5);
			\node at (3, 2.75) {$c_{6}$};
			\filldraw[fill=blue!10!white, draw=black] (0,3.0) rectangle ++(6.0,0.5);
			\node at (3, 3.25) {$c_{7}$};
			\filldraw[fill=blue!10!white, draw=black] (0,3.5) rectangle ++(6.0,0.5);
			\node at (3, 3.75) {$c_{8}$};
			\filldraw[fill=blue!10!white, draw=black] (0,4.0) rectangle ++(6.0,0.5);
			\node at (3, 4.25) {$c_{9}$};
			\filldraw[fill=blue!10!white, draw=black] (0,4.5) rectangle ++(6.0,0.5);
			\node at (3, 4.75) {$c_{10}$};
			
			\filldraw[fill=blue!10!white, draw=black] (6.0,2.0) rectangle ++(2.0,0.5);
			\node at (7, 2.25) {$c_{11}$};
			\filldraw[fill=blue!10!white, draw=black] (6.0,2.5) rectangle ++(2.0,0.5);
			\node at (7, 2.75) {$c_{12}$};
		\end{tikzpicture}
	\end{center}
	This instance is not laminar (because of the two voters not approving $X$), but it is close to being laminar and it is reasonable to expect that the elected committee should be the same as the green one from \Cref{ex:gcr_not_laminar}. Equal Shares uniquely elects that committee. On the other hand, GCR selects first the weakly $(6, Y)$-cohesive group containing the first $3000$ voters and in the second step the weakly $(2,Z)$-cohesive group containing the last $1000$ voters. After that the algorithm stops, electing committee $Y\cup Z$, as depicted above. Note that in this case, the choice of weakly cohesive groups is unique. \qed
\end{example}

\Cref{ex:gcr_not_laminar,ex:gcr_not_nearly_laminar} do not rule out the existence of an FJR rule that is also laminar proportional; the existence of a natural example of such a rule is an interesting open problem.

\section{Equal Shares and GCR for Ordinal Ballots}
\label{sec:ordinal_model}

In this section we discuss how our two rules can be adapted for committee elections where voters have \myemph{ordinal preferences}, that is, voters express their preferences by ranking the candidates. 
The main idea is straightforward: we convert voters' preference rankings into additive valuations, by using positional scoring rules, and then apply our rules to the resulting election. Note that if we use scoring rules that assign positive values to positions in voters' rankings, we always obtain exhaustive rules. We will show that when we use a lexicographic conversion scheme in which voters care infinitely more about their top-ranked candidate than their second-ranked candidate and so on, then Equal Shares satisfies an axiom called Proportionality for Solid Coalitions (PSC) which was first introduced to analyze the Single Transferable Vote \citep{woo:j:properties,tid-ric:j:stv}. (GCR does not satisfy PSC.) The Method of Equal Shares as applied to ordinal preferences is an example of an Expanding Approvals rule, a class of rules identified by \citet{az-bar:expanding_approval}.  Interestingly, due to its flexibility, Equal Shares can be used to extend the proportionality idea behind PSC beyond a lexicographic interpretation of preferences: Depending on how we convert voters' preference rankings to utilities, we obtain different forms of proportionality (cf., \citealp{fal-sko-szu-tal:stv_and_pav}). For example, roughly speaking, if we use Borda scores, the rule chooses outcomes where (due to EJR1) for all groups $S$ of voters that rank a set $T$ of candidates highly, there is a voter in $S$ for whom enough highly-ranked candidates are included in the committee.

\subsection{Model for Ordinal Preferences}

In this section we assume that each voter $i \in N$ submits a strict preference order $\succ_i$ over the set of candidates. The order $c_{i_1} \succ_i c_{i_2} \succ \ldots \succ c_{i_m}$ means that $c_{i_1}$ is voter $i$'s most preferred candidate, $c_{i_2}$ is her second most preferred candidate, and so on. By $\pos_i(c)$ we denote the position of candidate $c$ in $i$'s preference ranking. In the above example we have $\pos_i(c_{i_1}) = 1$, $\pos_i(c_{i_2}) = 2$, and so on. For sets $A$ and $B$, we write $A \succ_i B$ if $a \succ_i b$ for all $a \in A,b \in B$.

Further, we assume unit costs, so that the goal is to select a committee of $b$ candidates.

\begin{definition}[Proportionality for Solid Coalitions (PSC)]\label{def:psc}
	An outcome $W$ satisfies \myemph{PSC} if for each $\ell \in [b]$, each subset of voters $S \subseteq N$ with $|S| \geq \ell \cdot \frac{n}{b}$, and each subset of candidates $T$ such that $T \succ_i C \setminus T$ for all $i \in S$, it holds that $|W \cap T| \geq \min(\ell, |T|)$.
	
	A rule satisfies PSC if for each election it only returns outcomes that satisfy PSC.
\end{definition}

\subsection{Equal Shares and PSC}

The definition of PSC focuses on voters' top preferences: it gives guarantees based on voters' top-1 preferences, also their top-2 preferences, and so on. Thus, intuitively, it is based on an interpretation where if $c \succ_i c'$, then the utility that voter $i$ assigns to candidate $c$ is infinitely higher than that assigned to $c'$. Equal Shares naturally extends to such preferences, which we call \emph{lexicographic utilities}, and we will see the rule then satisfies PSC. To implement lexicographic utilities, we need to adapt \Cref{def:rulex} to use a slightly different interpretation of the price per unit of utility, $\rho$. Normally, we assume that $\rho$ is a positive real number; in order to adapt the definition to lexicographic preferences we let $\rho$ be a pair, so that $\rho = (\rho_{\text{pos}}, \rho_{\text{pay}}) \in [m] \times \mathbb R_{\ge 0}$. Values of $\rho$ are lexicographically ordered, so $\rho < \rho'$ if and only if either $\rho_{\text{pos}} < \rho'_{\text{pos}}$ or $\rho_{\text{pos}} = \rho'_{\text{pos}}$ and $\rho_{\text{pay}} < \rho'_{\text{pay}}$. To define Equal Shares, we do not need to specify values for the utilities $u_i(c)$, but we only need to define how to multiply $\rho$ and $u_i(c)$, which we do as follows:
\begin{align*}
	\rho \cdot u_i(c) =
	\begin{cases} 
		\infty & \text{if } \pos_i(c) < \rho_{\text{pos}}, \\
		\rho_{\text{pay}} & \text{if } \pos_i(c) = \rho_{\text{pos}}, \\
		0 & \text{if } \pos_i(c) > \rho_{\text{pos}}. \\
	\end{cases} 
\end{align*}
With these definitions in place, we can evaluate the Method of Equal Shares for lexicographic utilities, induced by a profile of ordinal preferences.

Intuitively, when considering how to fund a candidate $c$, this rule looks for the highest position $r = \rho_{\text{pos}} \in [m]$ such that it is possible to cover $\cost(c)$ by taking all the remaining money of all voters who rank $c$ in positions $1, \dots, r-1$, and taking from voters who rank $c$ exactly in position $r$ either $\rho_{\text{pay}}$ or all their remaining money. Voters who rank $c$ in a lower position than $r$ do not pay.

Defined this way, the Method of Equal Shares is an example in the class of \emph{Expanding Approvals} rules \citep{az-bar:expanding_approval,aziz2021proportionally}. All such rules satisfy PSC \citep{aziz2021proportionally}, and thus Equal Shares satisfies PSC. We include a simple proof for completeness.

\begin{proposition}\label{prop:rulex_and_psc}
	The Method of Equal Shares for lexicographic utilities satisfies PSC.
\end{proposition}
\begin{proof}
	Consider a committee $W$ returned by Equal Shares for an election instance $(N, C, b, \{\succ_i\}_{i \in N})$. Let $\ell \in [b]$, $S \subseteq N$, and $T$ be as in the definition of PSC (\Cref{def:psc}). We need to show that $|W \cap T| \ge \min(\ell, |T|)$. The voters in $S$ initially have the following budget:
	\begin{align*}
		\textstyle
		|S| \cdot \frac{b}{n} \geq \ell \cdot \frac{n}{b} \cdot \frac{b}{n} = \ell \text{.}
	\end{align*}
	
	Consider the steps of Equal Shares as the value of $\rho_{\text{pos}}$ increases from $1$ to $|T|$. In each such step, each voter from $S$ can pay only for the candidates in $T$. Indeed, each candidate $c \in C \setminus T$ occupies a worse position than $|T|$ in those voters' preference rankings, and so for each $i \in S$ we have $\rho \cdot u_i(c) = 0$ (since $\pos_i(c) \ge |T| > \rho_{\text{pos}}$). When $\rho_{\text{pos}}$ reaches $|T|$, then for each candidate $c \in T$ and each $i \in S$ we have $\rho \cdot u_i(c) \ge \rho_{\text{pay}} > 0$. If at this point, we have $|W \cap T| < \min(\ell, |T|)$, then there is some $c \in T$ that has not yet been added to $W$, so $c \not\in W$. Also the voters in $S$ have paid money to add at most $\ell - 1$ candidates to $W$. Thus, they together have at least a budget of $\ell - (\ell -1) = 1$ left. Hence Equal Shares has not yet terminated, because it could add $c$ to the outcome. Because Equal Shares will terminate, at the end it must be true that $|W \cap T| \ge \min(\ell, |T|)$, as desired.
\end{proof}

One may wonder, given the lexicographic utility scheme, whether PSC is just a consequence of Equal Shares satisfying EJR. \Cref{ex:ejr_and_psc} below shows that this is not the case and that the two axioms are logically incomparable in this context. FJR and PSC are also logically incomparable.

\begin{example}[PSC is logically incomparable to both EJR and FJR]\label{ex:ejr_and_psc}
	Consider three voters with the following preference orders over $C = \{c_1, c_2, c_3, c_4\}$:
	\begin{align*}
		1\colon & c_1 \succ c_2 \succ c_3 \succ c_4 \\
		2\colon & c_2 \succ c_3 \succ c_1 \succ c_4 \\
		3\colon & c_3 \succ c_1 \succ c_2 \succ c_4 \text{.}
	\end{align*}
	Assume $b = 2$. In this example PSC would require that two candidates from $\{c_1, c_2, c_3\}$ are elected. On the other hand, committee $\{c_1, c_4\}$ satisfies FJR. Thus, EJR does not imply PSC. Because FJR is stronger than EJR, it follows that FJR does not imply PSC.
	
	Consider two voters with the following preferences:
	\begin{align*}
		1\colon & c_1 \succ c_2 \succ c_3 \succ c_4 \\
		2\colon & c_4 \succ c_1 \succ c_3 \succ c_2 \text{.}
	\end{align*}
	Assume $b=1$. Here, EJR requires that $c_1$, $c_2$, or $c_4$ must be selected. On the other hand, $\{c_3\}$ is a committee that satisfies PSC. Thus, PSC does not imply EJR. Because FJR is stronger than EJR, it follows that PSC does not imply FJR.
	\qed
\end{example}

\subsection{GCR and PSC}

GCR can also be adapted to lexicographic utilities. In this case, it is sufficient to assume that the utilities are exponentially decreasing with the positions: for each $i \in N$ and $c \in C$ we set $u_i(c) = m^{-\pos_i(c)}$.\footnote{This is sufficient because in GCR, utilities are only used to compare sets of candidates in order of preference. For Equal Shares, we needed to use ``infinitesimal'' values because utility numbers are also used to decide how much each voter pays.} Then, for each $c$ we have that $u_i(c) > \sum_{c' \prec_i c} u_i(c')$, and so the utility a voter assigns to a candidate in position $p$ is higher than the utility that it would assign to any committee all of whose members are ranked below $p$. 

\begin{example}[GCR for lexicographic utilities fails PSC]
	Consider the following preference profile:
	\begin{align*}
		1\colon & c_1 \succ c_7 \succ c_8 \succ c_6  \succ c_4 \succ c_5 \succ c_2 \succ c_3 \succ c_9 \succ c_{10} \succ c_{11} \succ c_{12} \\
		2\colon & c_1 \succ c_7 \succ c_8 \succ c_6  \succ c_4 \succ c_5 \succ c_2 \succ c_3 \succ c_9 \succ c_{10} \succ c_{11} \succ c_{12} \\
		3\colon & c_1 \succ c_2 \succ c_3 \succ c_6  \succ c_4 \succ c_5 \succ c_7 \succ c_8 \succ c_9 \succ c_{10} \succ c_{11} \succ c_{12} \\
		4\colon & c_1 \succ c_2 \succ c_3 \succ c_6  \succ c_4 \succ c_5 \succ c_7 \succ c_8 \succ c_9 \succ c_{10} \succ c_{11} \succ c_{12} \\
		5\colon & c_1 \succ c_2 \succ c_3 \succ c_6  \succ c_4 \succ c_5 \succ c_7 \succ c_8 \succ c_9 \succ c_{10} \succ c_{11} \succ c_{12} \\
		6\colon & c_1 \succ c_2 \succ c_3 \succ c_6  \succ c_4 \succ c_5 \succ c_7 \succ c_8 \succ c_9 \succ c_{10} \succ c_{11} \succ c_{12} \\
		7\colon & c_2 \succ c_3 \succ c_1 \succ c_7  \succ c_8 \succ c_4 \succ c_5 \succ c_6 \succ c_9 \succ c_{10} \succ c_{11} \succ c_{12} \\
		8\colon & c_3 \succ c_2 \succ c_1 \succ c_7  \succ c_8 \succ c_4 \succ c_5 \succ c_6 \succ c_9 \succ c_{10} \succ c_{11} \succ c_{12} \\
		9\colon & c_4 \succ c_5 \succ c_9 \succ c_7  \succ c_8 \succ c_1 \succ c_2 \succ c_3 \succ c_6 \succ c_{10} \succ c_{11} \succ c_{12} \\
		10\colon & c_5 \succ c_4 \succ c_9 \succ c_7  \succ c_8 \succ c_1 \succ c_2 \succ c_3 \succ c_6 \succ c_{10} \succ c_{11} \succ c_{12} \\
		11\colon & c_{10} \succ c_{11} \succ c_{12} \succ c_7  \succ c_8 \succ c_1 \succ c_2 \succ c_3 \succ c_4 \succ c_5 \succ c_6 \succ c_9 \\
		12\colon & c_{11} \succ c_{10} \succ c_{12} \succ c_7  \succ c_8 \succ c_1 \succ c_2 \succ c_3 \succ c_4 \succ c_5 \succ c_6 \succ c_9 \text{.}
	\end{align*}
	
	Assume $b = 4$. Here, GCR will first pick $S = \{1, \ldots, 6\}$ as a weakly cohesive group, with the corresponding set of candidates $T = \{c_1, c_6\}$. Indeed, if $T$ consisted of $3$ candidates, then $S$ would need to have at least 9 voters. However, any 9 voters rank at least 4 different candidates at the top position, thus at least one of them would have a lower satisfaction than the voters from $S$ have from $T$. By the same argument, $T$ cannot consist of 4 candidates. If $T$ consisted of 2 candidates but $S$ included one voter from $7, \ldots, 12$, then the satisfaction of voter 1 or 2 would also be lower. Indeed, these two voters rank $c_2$, $c_3$, $c_4$, $c_5$, $c_{10}$, and $c_{11}$ (that is candidates that appear in the top positions) below $c_6$.
	
	Hence, GCR picks $c_1$ and $c_6$, and labels the first 6 voters as inactive. In the second step, the rule picks $c_7$ and $c_8$. This is because each other candidate appears at most twice before $c_7$ and $c_8$ in the remaining voters' rankings. Thus, the rule picks $c_1, c_6, c_7$ and~$c_8$.
	
	On the other hand, by looking at voters $3, \ldots, 8$ we observe that PSC requires that two candidates from $c_1, c_2, c_3$ should be selected.
	\qed
\end{example}

\section{Conclusion}

In this paper, we have formulated two axioms, EJR and FJR, that capture the idea of proportionality in the participatory budgeting (PB) model. We have designed a simple and natural rule for the PB model, the Method of Equal Shares. It satisfies EJR and other proportionality-related properties, and it is computable in polynomial time. The stronger of our two properties, FJR, is also satisfiable, albeit by a different and arguably less natural voting rule. It is an interesting open question whether there exists a natural voting rule that satisfies FJR and shares other desirable properties of Equal Shares.  

Many cities run PB by dividing the overall budget between districts and running separate elections in each district. In particular, this is true in the Polish cities that provide our experimental election data. We claimed in the introduction that this practice of separate elections leads to inferior outcomes. We designed a final experiment to study this question. Our results show a visible advantage of using global rules such as Equal Shares over separate district elections. For example, Equal Shares always produces outcomes with a more equal distribution of voter utility, and in most cases also provides a higher total utility in comparison to the rules that are in actual use in the elections we examined. 

In the conference proceedings version of this paper \citep[\href{https://proceedings.neurips.cc/paper/2021/hash/69f8ea31de0c00502b2ae571fbab1f95-Abstract.html}{link to online version}]{peters2021neurips}, we also report the results of some experiments on data from the \href{http://pabulib.org/}{Pabulib.org} library \citep{stolicki2020pabulib}, which contains full voting data from large PB elections in several Polish cities including Warsaw and Krak\'ow. In those experiments, we compare the outputs of different rules (Equal Shares, Phragm\'en's rule, and a PAV-like rule), and we compare the behavior of different completions of Equal Shares.

\subsection*{Acknowledgments}
Grzegorz Pierczy\'nski and Piotr Skowron were supported by Poland's National Science Center grant UMO-2019/35/B/ST6/02215. Part of this work was done while Dominik Peters was at Carnegie Mellon University and Harvard University.

\bibliographystyle{plainnat}

\end{document}